\PassOptionsToPackage{colorlinks=true, allcolors=blue}{hyperref}
\documentclass[sigconf]{acmart}
\usepackage[linesnumbered,ruled,vlined]{algorithm2e} 
\usepackage{tikz}
\usepackage{graphicx}
\usepackage{amsthm}
\usepackage{subfigure}
\usepackage{float} 
\usepackage{makecell}
\usepackage{colortbl}
\usepackage{caption}
\usepackage{amsmath}
\usepackage{graphicx}
\usepackage{enumitem}

\newtheorem{problem}{Problem}
\newtheorem{definition}{Definition}

\newtheorem{lemma}{Lemma}

\newtheorem{proposition}{Proposition}

\begin{document}
\title{Maximum Edge-based Quasi-Clique: Novel Iterative Frameworks}
\author{Hongbo Xia}
\affiliation{
    \institution{Harbin Institute of Technology, Shenzhen}
    \city{Shenzhen}
    \country{China}
}
\email{24s151167@stu.hit.edu.cn}

\author{Shengxin Liu}
\affiliation{
    \institution{Harbin Institute of Technology, Shenzhen}
    \city{Shenzhen}
    \country{China}
}
\email{sxliu@hit.edu.cn}
\authornote{Shengxin Liu and Zhaoquan Gu are the corresponding authors.}

\author{Zhaoquan Gu}
\affiliation{
    \institution{Harbin Institute of Technology, Shenzhen}
  \country{}
}
\affiliation{
    \institution{Pengcheng Laboratory}
    \city{Shenzhen}
    \country{China}
}
\email{guzhaoquan@hit.edu.cn}
\authornotemark[1]

\begin{abstract}
Extracting cohesive subgraphs from complex networks is a fundamental task in graph analytics and is essential for understanding biological, social, and web graphs. The edge-based $\gamma$-quasi-clique model offers a flexible alternative by identifying subgraphs whose edge densities exceed a specified threshold $\gamma$. However, finding the exact maximum edge-based quasi-clique is computationally challenging, as the problem is NP-hard and lacks the hereditary property. These characteristics limit the effectiveness of conventional pruning methods and the development of efficient reduction rules. As a result, existing algorithms, such as \texttt{QClique} and \texttt{FPCE}, struggle to scale to large graphs. In this paper, we revisit the problem and propose a novel iterative framework that reformulates the problem as a sequence of hereditary subproblems, enabling more effective pruning and reduction strategies and improving the worst-case time complexity. Furthermore, we redesign the iterative process and introduce a novel heuristic to further improve practical efficiency. Extensive experiments on 253 large-scale real-world graphs demonstrate that our proposed algorithm \texttt{EQC-Pro} outperforms existing methods by up to four orders of magnitude.
\end{abstract}

\ccsdesc[500]{Theory of computation~Graph algorithms analysis}

\keywords{Cohesive subgraph mining, edge-based quasi-clique}

\maketitle
\section{Introduction}
Many types of real-world information can be naturally represented as graphs, including biological protein networks, social networks, and web graphs. A fundamental problem in this context is the extraction of \emph{cohesive subgraphs}. However, due to noise and incomplete edge information, the strict definition of a clique (i.e., a pairwise fully connected subgraph) is often too restrictive for practical scenarios. 
One intuitive relaxation is the \emph{edge-based \(\gamma\)-quasi-clique}~\cite{abello1998massive,abello2002massive}, which introduces a parameter (between 0 and 1) to control the degree of relaxation. Specifically, a subgraph is considered an \emph{edge-based $\gamma$-quasi-clique} if the ratio of its number of edges to that of a complete graph of the same size is at least $\gamma$. When $\gamma = 1$, this definition reduces to the classical clique model. 

Recent studies~\cite{rahman2024pseudo} have provided strong empirical evidence that the edge-based $\gamma$-quasi-clique model offers practical benefits in biological network analysis. For instance, compared to the degree-based $\gamma$-quasi-clique model, which requires each vertex in the subgraph to be connected to at least a $\gamma$-fraction of the other vertices, the edge-based model achieves better predictive performance in identifying protein complexes, as measured by sensitivity, precision, and geometric accuracy~\cite{nepusz2012detecting}. Moreover, compared to heuristic approaches~\cite{nepusz2012detecting}, clusters obtained from exact solutions exhibit higher biological relevance, as demonstrated by Gene Ontology enrichment\footnote{https://geneontology.org/} analysis, which reveals enrichment in more unique and functionally meaningful cellular components. These results highlight the practical advantages of exact dense subgraph mining for uncovering biologically significant structures. Beyond biological networks, the edge-based quasi-clique model has also proven effective in a variety of other contexts, including the analysis of social networks~\cite{Bedi2016CommunityDetection,Fang2020cosub} and the detection of anomalies~\cite{tanner2010koobface,yu2021graph}.

In this paper, we revisit the problem of \emph{finding the exact maximum edge-based quasi-clique}~\cite{veremyev2016exact,ribeiro2019mqc,marinelli2020lp}. Specifically, given a graph and a density threshold $\gamma$, our goal is to identify the largest edge-based $\gamma$-quasi-clique. This problem is challenging as it has been proven to be NP-hard~\cite{pattillo2013maximum}, which creates substantial computational difficulties for exact solutions. Moreover, the edge-based $\gamma$-quasi-clique does not have the hereditary property; that is, its subgraphs are not necessarily edge-based $\gamma$-quasi-cliques. Consequently, many techniques that have been developed for finding cohesive subgraphs with the hereditary property, such as $k$-plexes~\cite{chang2022efficient,chang2024maximum,gao2024maximum} and $k$-defective cliques~\cite{chang2023kdefective,Dai2023defect,Chang24kDC-2}, cannot be directly applied to this problem. This limitation reduces the effectiveness of traditional pruning strategies and further increases the complexity of the problem.

\noindent \underline{\textbf{State-of-the-Art Algorithms.}}
Existing exact algorithms for the maximum edge-based quasi-clique problem fall into two main categories: mixed integer programming (MIP) and branch-and-bound (BNB) techniques. While MIP-based methods~\cite{pattillo2013maximum,marinelli2020lp,veremyev2016exact} can solve instances with up to $10^4$ vertices, they often struggle to scale to larger graphs in practice. Consequently, recent research has focused on developing efficient BNB algorithms.

Among BNB methods, \texttt{QClique}~\cite{ribeiro2019mqc} stands out as the most efficient exact solver to date. It builds on the branching strategy of \texttt{PCE}~\cite{Uno2010}, which exploits the quasi-hereditary structure of edge-based quasi-cliques~\cite{pattillo2013maximum}. \texttt{QClique} introduces a new upper bound to improve pruning, achieving performance competitive with earlier BNB and MIP-based approaches. In parallel, \texttt{FPCE}~\cite{rahman2024pseudo} represents the state-of-the-art algorithm for maximal enumeration and incorporates several refined pruning rules. Although \texttt{FPCE} is designed for enumerating maximal edge-based quasi-cliques, it can be adapted to the exact maximum problem by dynamically tracking the current best solution during search.

Despite these advances, both \texttt{QClique} and \texttt{FPCE} share certain limitations. Specifically, they inherit the same branching mechanism as \texttt{PCE}, making their pruning strategies tightly coupled to this structure. This coupling hinders the design of effective reduction rules for further efficiency improvements; to the best of our knowledge, no such reduction rules have been proposed in the literature. As a result, these solvers remain constrained on large graphs, where graph reduction is often decisive for performance.

\noindent \underline{\textbf{Our New Methods.}} We reformulate the problem as a sequence of hereditary subproblems, enabling pruning and reduction techniques that explicitly leverage hereditary properties. Under this reformulation, we formally prove the correctness of the proposed iterative framework for the edge-based quasi-clique and establish its worst-case time complexity of $O^*(\beta_\kappa^n)$, where $O^*$ suppresses polynomial factors, $n$ is the number of vertices, $\beta_\kappa < 2$ is the largest real root of $x^{\kappa+2} - 2x^{\kappa+1} + 1 = 0$, and $\kappa$ relates to the size of the optimal solution. This improves upon existing exact algorithms such as \texttt{QClique} and \texttt{FPCE}, both of which have $O^*(2^n)$ complexity.
To the best of our knowledge, \texttt{EQC-Pro} is the first exact algorithm for \texttt{MaxEQC} with a non-trivial complexity.

Furthermore, by exploiting the quasi-hereditary property of the edge-based quasi-clique, we develop a bottom-up doubling iterative framework that reduces the number of iterations required to solve hereditary subproblems from $O(n)$ to $O(\log s^*)$ in theory, where $s^*$ denotes the size of the optimal solution and is typically much smaller than $n$ in practice. 
Finally, we design a simple yet effective heuristic grounded in the degeneracy sequence, a standard tool in cohesive subgraph mining. To mitigate the rigidity induced by strict adherence to a fixed degeneracy sequence, we introduce a dynamically refined heuristic strategy that iteratively updates the candidate solution via vertex additions and deletions. This strategy diversifies selection process and can yield better heuristic solutions.

\noindent \underline{\textbf{Contributions.}} Our contributions are summarized as follows.
\begin{itemize}[leftmargin=*, topsep=0.3em]
    \item We adapt the top-down iterative framework to the edge-based quasi-clique setting by reducing the problem to a sequence of hereditary subproblems, thereby achieving improved complexity of \(O^*(\beta_\kappa^n)\). (Section~\ref{sec:top-down})
    \item By leveraging the quasi-hereditary property of edge-based quasi-clique, we propose a bottom-up doubling iterative framework that reduces the number of iterations required to solve hereditary subproblems from \(O(n)\) to \(O(\log s^*)\). (Section~\ref{sec:bottom-up})
    \item We design a novel heuristic strategy, further improving practical performance. (Section~\ref{sec:heu})
\end{itemize}
We conduct extensive experiments on 253 graphs and show that our proposed \texttt{EQC-Pro} achieves up to four orders of magnitude speedup over state-of-the-art baselines \texttt{QClique} and \texttt{FPCE} (Section~\ref{sec:exp}). Our source code can be found in~\cite{appendix}.

\section{Preliminaries}\label{sec:prelim}
Let $G=(V, E)$ be an undirected simple graph with $|V| = n$ vertices and $|E|=m$ edges.
We denote $G[S]$ of $G$ as the subgraph induced by the set of vertices $S \subseteq V$.
For any \(S \subseteq V\) and \(u \in V\), let \(d_G^S(u)\) denote the number of neighbors of \(u\) in \(S\), i.e., $d_G^S(u) = \left| \{ v \in S \mid v \ne u \text{ and } (u,v) \in E \} \right|$.
Moreover, we use \( N_G[u] \) to denote the closed neighborhood of vertex \( u \) in \( G \), i.e., $N_G[u] = \{ v \in V \mid v = u \text{ or } (u,v) \in E \}$, and \( N_G[S] \) to denote the union of \( N_G[u] \) for all vertices \( u \in S \), i.e., $N_G[S] = \bigcup_{u \in S} N_G[u]$.
For an induced graph $g$ of $G$, the vertex set and the edge set of $g$ are denoted by $V(g)$ and $E(g)$, respectively.
For a subgraph \(g\), we define \emph{the number of missing edges} as the number of edges absent in \(g\) compared to a complete graph of the same size, i.e., \( \binom{|V(g)|}{2} - |E(g)|\).
In this paper, we focus on the edge-based $\gamma$-quasi-clique ($\gamma$-EQC)~\cite{abello1998massive,abello2002massive}.
\begin{definition}
Given a real number \( 0 < \gamma < 1 \), a graph \( g \) is said to be an \emph{edge-based \( \gamma \)-quasi-clique} (\( \gamma \)-EQC) if $|E(g)| \geq \gamma \times \binom{|V(g)|}{2}$.
\end{definition}
In other words, an edge-based $\gamma$-quasi-clique is an induced subgraph in which the number of edges is at least a $\gamma$ faction of the number of edges in the corresponding complete graph, i.e., $\binom{|V(g)|}{2}$. We are ready to present our problem in this paper.
\begin{problem}[Maximum edge-based $\gamma$-quasi-clique]
Given a graph $G=(V, E)$ and a real number $\gamma \in (0, 1)$, the Maximum Edge-based $\gamma$-Quasi-Clique Problem (\texttt{MaxEQC}) aims to find an edge-based $\gamma$-quasi-clique in $G$ with the largest number of vertices.
\end{problem}
Let $g^*$ and $s^*=|V(g^*)|$ be the maximum $\gamma$-EQC and its size. Finding the maximum $\gamma$-EQC is known to be NP-hard~\cite{pattillo2013maximum}. 
Moreover, $\gamma$-EQCs lack the \emph{hereditary property}~\cite{pattillo2013maximum}: A subgraph of a $\gamma$-EQC is not necessarily a $\gamma$-EQC. This limitation complicates efficient algorithm design, as even checking maximality by inclusion becomes a non-trivial task for $\gamma$-EQC. Nevertheless, $\gamma$-EQCs satisfy a weaker structural guarantee, known as the \emph{quasi‐hereditary property}.

\begin{proposition}[Quasi‐hereditary property~\cite{pattillo2013maximum}] \label{pro:quasi-hereditary}
Let $G=(V,E)$ be a graph and $S\subseteq V$ such that $G[S]$ is a $\gamma$-EQC. If there exists a vertex $v\in S$ with degree less than the average degree in $G[S]$, i.e., $d_G^S(v) \;\le\; \frac{2\,|E(G[S])|}{|S|}$, then $G[S']$, where $S' = S\setminus\{v\}$, is a $\gamma$‐EQC.
\end{proposition}
Proposition~\ref{pro:quasi-hereditary} implies that for a given $\gamma$-EQC $g$, removing a vertex satisfying a specific degree constraint guarantees that the remaining graph remains a $\gamma$-EQC, whereas removing an arbitrary vertex does not, due to the non-hereditary property of $\gamma$-EQCs.

\section{A Top-down Iterative Baseline Framework}\label{sec:top-down}
A closely related problem to \texttt{MaxEQC} is the \emph{maximum degree-based quasi-clique problem (\texttt{MaxQC})}, which has been effectively addressed by the state-of-the-art iterative algorithm \texttt{IterQC}~\cite{xia2025iterqc}. We note that neither of these two problems satisfies the hereditary property. Thus, inspired by the key iterative idea behind \texttt{IterQC}, our first framework seeks to transform the \texttt{MaxEQC} problem into a sequence of subproblems that do possess the hereditary property, following a top-down approach.
We introduce the cohesive subgraph of $k$-defective clique, which serves as a key component in our approach.

\begin{definition}[\cite{yu2006defective}]\label{def:defective}
Given an integer $k\geq 1$, a graph $g$ is said to be a {\em $k$-defective clique} if $|{E}(g)|\ge \binom{|V(g)|}{2}-k$.
\end{definition}
In other words, a \(k\)-defective clique \(g\) contains at most \(k\) missing edges. Note that the $k$-defective clique enjoys the hereditary property, which can be leveraged to prune the search space.
Building on the concept of $k$-defective clique, extensive studies~\cite{chen2021computing,gao2022exact,Luo24defective,chang2023kdefective,Chang24kDC-2,dai2024theoretically} have focused on the following problem. 

\begin{problem}[Maximum $k$-defective clique]\label{def:kdc}
Given a graph $G=(V, E)$ and a positive integer $k$, the Maximum $k$-Defective-Clique Problem (\texttt{kDC}) aims to find the largest $k$-defective-clique in $G$. 
\end{problem}

Based on \texttt{kDC}, we can define the following two functions.
\begin{itemize}[leftmargin=*, topsep=0.2em]
        \item \texttt{get-k}$(n) := \lfloor (1 - \gamma) \cdot \binom{n}{2}\rfloor$, which takes a number $n$ of vertices as input and returns an appropriate value of $k$;
        \item \texttt{solve-defect}$(k) :=$ the size of the largest $k$-defective clique in $G$, which takes a value $k$ as input.
\end{itemize}
Inspired by \texttt{IterQC}~\cite{xia2025iterqc}, our top-down iterative approach \texttt{EQC-TD} starts from an upper bound $s_0 > s^*$ and alternately applies $k_{i} = \texttt{get-k}(s_{i-1})$ and $s_{i} = \texttt{solve-defect}(k_{i})$ 
until convergence, i.e., until $s_{i} = s_{i-1}$. 
It can be shown that the sequence $s_i$ is monotonically decreasing and converges to the optimal value $s^*$. 
Thus, \texttt{EQC-TD} correctly reduces the \texttt{MaxEQC} problem, which lacks the hereditary property, to a sequence of \texttt{kDC} instances, which are hereditary. Due to space constraints, the pseudocode and correctness proof of \texttt{EQC-TD} are provided in Appendix~\ref{sec:top-down-appendix}.

\noindent \textbf{\underline{Limitations.}} 
 The framework still suffers from the efficiency issues. \textbf{First}, in the worst case, \texttt{EQC-TD} requires \( n \) calls to \texttt{solve-defect},
 leading to excessive number of iterations. \textbf{Second}, although \texttt{kDC} is hereditary, its computational cost becomes prohibitive as \(k\) increases. Specifically, for a fixed density \(\gamma\), the \texttt{get-k} function indicates that \(k\) grows quadratically with the number of vertices.

\section{An Advanced Bottom-up Iterative Framework}\label{sec:bottom-up}
To address the limitations of the top-down iterative framework, we leverage a structural property, i.e., the quasi-hereditary property in Proposition~\ref{pro:quasi-hereditary}, unique to the definition of $\gamma$-EQC, which is absent under the degree-based formulation. Based on this property, we develop a specially designed bottom-up iterative framework tailored for \texttt{MaxEQC}.
We first introduce an important property of the iterative framework for \(\gamma\)-EQC, which is derived from Proposition~\ref{pro:quasi-hereditary}.

\begin{proposition}\label{pro:bin-property}
Recall that $s^*$ denote the size of the maximum edge-based $\gamma$-quasi-clique. Then, it holds that\\
\textbf{(1)} for any $s \le s^*$, $s \le \texttt{solve-defect}(\texttt{get-k}(s))$, and \\
\textbf{(2)} for any $s > s^*$, $s > \texttt{solve-defect}(\texttt{get-k}(s))$.
\end{proposition}
\begin{proof}
We recall that the function \(\texttt{solve-defect}(k)\) returns the \emph{largest} integer \(s'\) such that there exists a \(k\)-defective clique of size \(s'\). We then consider two cases according to the relationship between $s$ and $s^*$. In both cases, we let $k = \texttt{get-k}(s) = \left\lfloor (1 - \gamma)\binom{s}{2} \right\rfloor$.

\noindent\underline{\textbf{Case 1: \(s \le s^*\).}}
We show that it is always possible to construct a $\gamma$-EQC of size $s$ from the maximum $\gamma$-EQC $g^*$ of size $s^*$, such that $s \le \texttt{solve-defect}(\texttt{get-k}(s))$ holds.
Specifically, we begin with $g^*$ and iteratively remove a vertex of minimum degree from the subgraph until only \(s\) vertices remain. Since a vertex of minimum degree always has less than or equal to the average degree, the quasi-hereditary property (Proposition~\ref{pro:quasi-hereditary}) ensures that the induced subgraph $g$ on the remaining $s$ vertices remains a $\gamma$-EQC. Thus, the number of edges in $g$ is at least \(\gamma \times \binom{s}{2}\), which implies that $g$ is a feasible \(k\)-defective clique of size $s$ since the number of missing edges is \(\lfloor\binom{s}{2}-\gamma \times \binom{s}{2}\rfloor=k\). Consequently, we have \(\texttt{solve-defect}(k) \ge s\).

\noindent\underline{\textbf{Case 2: \(s > s^*\).}}
We prove by contradiction. Suppose, to the contrary, that $s \le \texttt{solve-defect}(\texttt{get-k}(s))$.
Let \(s' = \texttt{solve-defect}(\texttt{get-k}(s))\), and let $g'$ be the corresponding $k$-defective clique returned. By definition of a $k$-defective clique (Definition~\ref{def:defective}), the number of edges in $g'$ satisfies $E(g') \ge \binom{s'}{2} - \texttt{get-k}(s)$. Since \(\texttt{get-k}(s) \le \texttt{get-k}(s') = \left\lfloor (1 - \gamma)\binom{s'}{2} \right\rfloor\), we have $E(g') \ge \binom{s'}{2} - \left\lfloor (1 - \gamma)\binom{s'}{2} \right\rfloor \ge \gamma \binom{s'}{2}$.
Thus, $g'$ corresponds to a $\gamma$-EQC of size \(s'\), and we know \(s' \ge s > s^*\), which contradicts the maximality of \(s^*\).
\end{proof}

\noindent \textbf{\underline{Bottom-up Framework Overview.}} Proposition~\ref{pro:bin-property} implies that we can identify $s^*$ by carefully analyzing the relationship between $s$ and $\texttt{solve-defect}(\texttt{get-k}(s))$. In particular, based on Proposition~\ref{pro:bin-property}, we can devise a more efficient iterative framework by performing a binary search on the candidate size $s$. Consequently, the number of iterations to determine $s^*$ can be reduced from $O(n)$ to $O(\log n)$, thereby improving efficiency. Nevertheless, it is important to note that the efficiency of \texttt{solve-defect}($k$) is highly dependent on the value of $k$. Specifically, the problem becomes easier to solve as $k$ decreases. This observation motivates us to explore alternative search strategies that can further reduce the computational cost, especially in the early stages of the iterative framework.

To this end, we make use of a \emph{bottom‑up doubling strategy} to progressively approach the optimal value $s^*$. Specifically, let $lb$ and $ub$ denote the current lower and upper bounds of $s^*$. The standard binary search approach would iteratively select the midpoint, $s = \big\lfloor ({lb} + {ub})/2 \big\rfloor$, as the candidate value of $s$ to evaluate at each iteration, updating the bounds based on the result. While this method is efficient in terms of the number of iterations, it may not always be efficient in practice, particularly when we need to invoke \texttt{solve-defect} with a large value of $k$ (e.g., when $ub$ is large).

In contrast, our bottom‑up doubling strategy begins by probing the values of $s$ in an exponentially increasing sequence $\{ {lb}+1, {lb}+2,lb+4,\dots\}$, doubling the step size at each iteration until the value exceeds $s^*$. After a candidate interval that must contain $s^*$ is identified, we then switch to a more localized search, i.e., we perform a binary search within this narrowed interval. This approach keeps the values of $k$ encountered in the early rounds as small as possible, which can potentially improve the overall efficiency of the bottom‑up doubling strategy.

\begin{algorithm}[t]
\SetKwBlock{FuncCheck}{Function \texttt{check}($s$):}{}
\caption{A Bottom-up Doubling Framework: \texttt{EQC-BU}}
\label{alg:bottom-up-framework}
\KwIn{A graph $G=(V,E)$ and a real value $0 < \gamma < 1$}
\KwOut{The size of the maximum $\gamma$-EQC in $G$}

$lb \gets 1$; $ub \gets |V|$\;

\tcp{Doubling Phase}
$gap \gets 1$\; 
\While{$lb + gap \le ub$ \textbf{\emph{and}} \texttt{check}($lb + gap$)}{
    $lb_{\text{new}} \gets lb + gap$\;
    $gap \gets 2 \times gap$\;
}

\tcp{Halving Phase}
$ub_{\text{new}} \gets  \min(lb + gap - 1,ub)$\;
\While{$lb_{\text{new}} \neq ub_{\text{new}}$}{
    $mid \gets \lceil (lb_{\text{new}} + ub_{\text{new}})/2 \rceil$\;
    \eIf{\texttt{check}($mid$)}{
        $lb_{\text{new}} \gets mid$\;
    }{
        $ub_{\text{new}} \gets mid - 1$\;
    }
}

\Return{$lb_{\text{new}}$}\;

\FuncCheck{
  $k \gets \texttt{get-k}(s)$\;
  $s' \gets \texttt{solve-defect}(k)$\;
  \Return{$s' \ge s$}\;
}
\end{algorithm}

\noindent \textbf{\underline{Algorithm.}}
Our bottom-up doubling iterative framework \texttt{EQC-BU} is detailed in Algorithm~\ref{alg:bottom-up-framework}. \texttt{EQC-BU} searches for the optimal value $s^*$ by dynamically adjusting the lower bound $lb$ and upper bound $ub$ of $s^*$. 
In Line 1, \texttt{EQC-BU} initializes $lb$ to 1 and $ub$ to $|V|$, as it holds that $1 \leq s^* \leq |V|$. The core idea of \texttt{EQC-BU} is to efficiently and progressively identify the search interval containing $s^*$ via the \texttt{check} subroutine (described in Lines 14-17). This subroutine evaluates whether a candidate value of \( s \) does not exceed the optimal solution size \( s^* \), as guaranteed by Proposition~\ref{pro:bin-property}. The evaluation process has two phases: a doubling phase and a halving phase.

\noindent \textbf{(1) Doubling Phase (Lines 2-5)}. The doubling phase initializes the step size $gap$ to 1 (Line 2) and iteratively calls \texttt{check}($lb + gap$) for increasing values of $gap$ if $lb+gap$ does not exceed the current upper bound $ub$ (Line 3). If \texttt{check} returns \texttt{true}, then by Proposition~\ref{pro:bin-property}, $lb + gap$ is feasible, which implies $lb + gap \leq s^*$. This value is then used to update the lower bound $lb_{\text{new}}$ (Line 4), which will be the starting point for the halving phase. Then, in Line 5, the algorithm doubles $gap$ to expand the interval upward, aiming to quickly approach or exceed $s^*$. This process repeats until either the feasibility check fails or $ub$ is reached in Line 3. Note that if the current candidate value of $lb+gap$ is infeasible, we know $s^* < lb + gap$, which will be used to update $ub_{\text{new}}$ in the halving phase.\\
\noindent \textbf{(2) Halving Phase (Lines 6-12)}. After the doubling phase terminates, the algorithm proceeds to a binary search (halving) phase to precisely locate the optimal value $s^*$. In this phase, the refined interval $[lb_{\text{new}}, ub_{\text{new}}]$ is set according to the final state after doubling (Line 6). The halving phase then performs a standard binary search within this narrowed interval. In each iteration, it tests the midpoint $mid = \lceil (lb_{\text{new}} + ub_{\text{new}})/2 \rceil$. If \texttt{check}($mid$) returns \texttt{true}, the lower bound $lb_{\text{new}}$ is updated to $mid$; otherwise, the upper bound $ub_{\text{new}}$ is set to $mid - 1$ (Lines 7-12). This binary refinement continues until $lb_{\text{new}} = ub_{\text{new}}$, at which point the optimal value $s^*$ is identified. Finally, \texttt{EQC-BU} returns $lb_{\text{new}}$ as $s^*$ in Line 13.

\begin{figure}[t]
    \centering
    \vspace{-0.1in}
    \subfigure[The top-down iterative strategy (Algorithm~\ref{alg:top-down-framework})]{
        \includegraphics[width=0.42\textwidth, trim=5.9cm 21.6cm 5cm 7cm, clip]{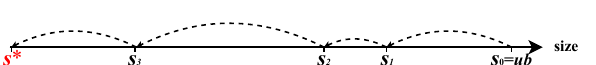}
        \label{fig:iteration-td}
    }
    \subfigure[The bottom-up doubling iterative strategy (Algorithm~\ref{alg:bottom-up-framework})]{
        \includegraphics[width=0.42\textwidth, trim=5.9cm 24.2cm 5cm 3.6cm, clip]{figures/Example/example_iteration.pdf}
        \label{fig:iteration—pro}
    }
    \vspace{-0.25in}
    \caption{Illustration of iterative frameworks.}
    \label{Fig:iteration}
    \vspace{-0.2in}
\end{figure}

\noindent \underline{\textbf{Comparison of Top-Down and Bottom-Up Strategies.}} To facilitate a clearer comparison between the top-down and bottom-up iterative frameworks, Figure~\ref{Fig:iteration} presents both strategies. Dashed arrows indicate updates to the upper bound, while solid arrows represent improvements to the lower bound.
Figure~\ref{fig:iteration-td} illustrates the top-down iterative strategy, in which it attempts to tighten the upper bound on the solution size. Dashed arrows represent these successive updates, which continue until the bound converges to the optimal size $s^*$. However, since each step does not guarantee a fixed reduction, the rate of progress may vary considerably across iterations. For example, there may be a substantial decrease from $s_2$ to $s_3$, or only a negligible improvement from $s_1$ to $s_2$.
The bottom-up doubling approach (Figure~\ref{fig:iteration—pro}) begins by rapidly identifying a solution interval through the doubling phase (shown in blue), and then performs the having phase (shown in black) within that interval to locate the optimal solution. Both approaches require ${O}(\log s^*)$ iterations in the worst case.

\noindent \textbf{\underline{Correctness and Remarks.}}
We remark that \texttt{EQC-BU} consists of a doubling phase, which quickly finds an upper bound of $s^*$, followed by a halving phase that efficiently narrows the candidate interval containing $s^*$. \texttt{EQC-BU} is specifically tailored for our iterative framework for the \texttt{MaxEQC} problem, where the computation performed by the invoked \texttt{kDC} solver becomes increasingly costly as $k$ grows.

For the correctness of \texttt{EQC-BU} (Algorithm~\ref{alg:bottom-up-framework}), we observe that, by Proposition~\ref{pro:bin-property}, \texttt{check}$(s)$ is monotonic in $s$: once it fails at some $s$, it fails for all larger values; conversely, once it succeeds, it succeeds for all smaller values. After the doubling phase (Lines 2-5), we observe that the optimal value $s^*$ always lies in the interval $[lb_\text{new},ub_\text{new}]$. This is because (1) $lb_\text{new}$ is only increased when a feasible candidate is found in Lines 2-5, and (2) $ub_\text{new} = \min(lb + gap - 1,ub)$ is set in Line 6 when $lb+gap > ub$ or an infeasible candidate is detected in Line 3. We then enter the halving (i.e., binary search) phase on this narrowed interval. In each iteration of the binary search, the midpoint test either raises the lower bound or lowers the upper bound while preserving the invariant. The process terminates precisely when the two bounds coincide, at which point the common value must be $s^*$. Thus, \texttt{EQC-BU} is correct.

We then analyze the number of calls to \texttt{check} in \texttt{EQC-BU} (and will discuss the time complexity in Section~\ref{sec:EQC-pro}).
The doubling phase (Lines 2-5) performs at most \(\lceil \log_2 s^* \rceil+1\) calls to \texttt{check}, since the algorithm starts from a small candidate size (e.g., 1) and doubles the size until it exceeds the actual optimal size \(s^*\). Once a valid upper bound is detected, where this upper bound is clearly at least \(s^*\) and at most \(2s^*\), the doubling phase terminates. In the subsequent halving phase (Lines 6-12), the search interval lies strictly below \(2s^*\) and is halved in each iteration. Thus, the number of iterations in this phase is also bounded by \(\lceil \log_2 s^* \rceil\). Overall, the total number of calls to \texttt{check} is bounded by \(O(\log s^*)\).

\section{Our Novel Heuristic Strategies} \label{sec:heu}
Since our bottom-up framework \texttt{EQC-BU} (Section~\ref{sec:bottom-up}) initiates iterations from a lower bound, i.e., a value smaller than the optimal value $s^*$, obtaining a larger lower bound in polynomial time through heuristics can reduce the total number of required iterations and thus improve overall efficiency. In this section, we propose novel heuristic strategies to compute a large lower bound for \texttt{MaxEQC}.

\begin{algorithm}[t]
    \caption{\texttt{Degen-Opt}$(G, \gamma)$}\label{alg:degen-opt}
    \KwOut{A large $\gamma$-EQC in $G$}
    $g \gets$ \texttt{Degen}$(G, \gamma)$\;
    \ForEach{vertex $u \in V$}{
        $g' \gets$ \texttt{Degen}$(G[N_G[u]], \gamma)$\;
        \lIf{$|V(g')| > |V(g)|$}{$g \gets g'$}
    }
    \Return{$g$}\;

    \SetKwFunction{FMain}{Degen}
    \SetKwProg{Fn}{Function}{:}{}
    \Fn{\FMain{$G$, $\gamma$}}{
        Compute a degeneracy ordering in $G$\;
        $H \gets$ the longest suffix of the degeneracy ordering such that $G[H]$ is a $\gamma$-EQC\;
        \Return{$G[H]$}\;
    }
\end{algorithm}

\subsection{Basic Heuristic Methods}
We present two basic heuristic methods \texttt{Degen} and \texttt{Degen-Opt} in Algorithm~\ref{alg:degen-opt}, both adapted from the heuristics in \texttt{KDC}~\cite{chang2023kdefective}. These methods utilize the concept of the \emph{degeneracy ordering}, which is obtained by iteratively removing the vertex with the smallest degree from the graph~\cite{BVZ03m}. The first basic heuristic algorithm \texttt{Degen} (Lines 6-9 of Algorithm~\ref{alg:degen-opt}) iteratively removes the vertex with the minimum degree until the remaining subgraph forms a $\gamma$-EQC. However, \texttt{Degen} has a structural limitation: once a vertex is selected into the solution, \emph{all subsequent vertices in the degeneracy ordering are inevitably included as well}. 
This inflexible selection mechanism results in a rigid combination process, thereby restricting the ability of the algorithm to explore more promising candidate subgraphs.

\texttt{Degen-Opt} in Algorithm~\ref{alg:degen-opt} partially relaxes this restriction. Specifically, for each vertex $u$ in the graph, we consider the subgraph induced by its closed neighborhood, i.e., $N_G[u] = \{ v \in V \mid v = u \text{ or } (u,v) \in E \}$, and execute \texttt{Degen} on this subgraph. In this way, each vertex has more opportunities to be included in the solution. Nevertheless, while \texttt{Degen-Opt} expands the combination space by constructing multiple sequences and thereby allows more flexibility in vertex selection, it still relies on fixed suffixes of degeneracy sequences. Moreover, since \texttt{Degen-Opt} is applied only to subgraphs induced by closed neighborhoods, the diameter of any resulting solution is at most 2, which is not a necessary condition for edge-based quasi-cliques. Thus, although \texttt{Degen-Opt} alleviates the rigidity of \texttt{Degen} to some extent, it does not resolve the limitation and additionally imposes an unnecessary restriction on solution diameter. 

\begin{algorithm}[t]
    \caption{\texttt{EQC-Heu}$(G, \gamma, g)$}\label{alg:subheu}
    \KwIn{A graph $G=(V,E)$, a real value $0 < \gamma < 1$, and a heuristic solution $g$}
    \KwOut{An improved $\gamma$-EQC in $G$}
    $g' \gets g; k \gets \texttt{get-k}(|V(g')|+1)$\; 
    \ForEach{vertex $u \in V(g)$}{
        $S \gets \{u\}$; $C \gets V \setminus \{u\}$\;
        \While{$|{E}(G[S])|\ge \binom{|S|}{2}-k$}{
            \If{$|S| > |V(g')|$}{
                $g' \gets G[S]$; $k \gets \texttt{get-k}(|V(g')|+1)$\;
            }
            \texttt{AddBestVertex}$(S,C)$;
            \texttt{AddBestVertex}$(S,C)$\;
            \texttt{RemoveWorstVertex}$(S,C)$\;
            \texttt{UpdateScore}$(S,C)$\;
        }
    }
    \Return{$g'$}\;

    \SetKwProg{Fn}{Function}{:}{}
\end{algorithm}

\subsection{Our Dynamic Refined Heuristic Strategies}
To further mitigate the inflexibility caused by the combination of \texttt{Degen-Opt}, we proposed our dynamic refined heuristic method \texttt{EQC-Heu-Pro}, which \emph{dynamically refines the heuristic solution by iteratively removing and adding vertices}. This strategy is designed to overcome the constraints on candidate solution construction.

\noindent \underline{\textbf{Algorithmic Overview.}} Our dynamic refined heuristic strategies take as input a feasible $\gamma$-EQC $g$ and iteratively adjust its vertex set to potentially obtain a larger heuristic solution. The core idea is to dynamically add and remove vertices between the selected set $S$ and the candidate set $C$, while maintaining the $\gamma$-EQC property throughout the process.
Specifically, for each vertex in $g$, we use it as a starting point (i.e., a \emph{seed}) and initialize the selected set with this seed. At each iteration, we update the selected set by \emph{adding two vertices and removing one vertex}, guided by a score function that aims to maximize the quality of the solution. After every update, we verify that the subgraph induced by the selected set remains a $\gamma$-EQC. This iterative process proceeds until the number of missing edges in the subgraph induced by the selected set exceeds the allowable threshold. In particular, given the largest $\gamma$-EQC $g'$ identified so far, we compute the maximum number of missing edges permitted for a $\gamma$-EQC of size $|V(g')|+1$, since our goal is to find a larger one.
After evaluating every vertex in $g$ as a seed, we select the largest $\gamma$-EQC found during this process as our final solution.

We note that the score function is designed to maximize the size of the heuristic solution by prioritizing the addition of vertices that increase the density of the subgraph induced by the selected set and the removal of those that decrease it. We adopt the scoring function from~\cite{wang2021nuqclq}, which is detailed in Appendix~\ref{sec:score_function_appendix} for completeness.

\noindent \underline{\textbf{\texttt{EQC-Heu}.}}
Our heuristic strategy \texttt{EQC-Heu} is summarized in Algorithm~\ref{alg:subheu}. We first initialize the current largest $\gamma$-EQC $g'$ and the maximum number $k$ of missing edges permitted for a $\gamma$-EQC of size $|V(g')|+1$ in Line 1 and consider each vertex $u$ in $V(g)$ as a seed in Lines 2-10.
For each $u$, we initialize the selected set $S \gets \{ u \}$ and the candidate set $C \gets V \setminus S$ in Line 3.
We then iteratively update $S$ and $C$ by adding two vertices from $C$ to $S$ (Line 7) and removing one vertex from $S$ to $C$ (Line 8), using \texttt{AddBestVertex} and \texttt{RemoveWorstVertex}, which select vertices according to their scores. Then the \texttt{UpdateScore} function computes the score for each vertex based on the current selected set $S$ and the candidate set $C$. 
(The details of \texttt{AddBestVertex}, \texttt{RemoveWorstVertex}, and \texttt{UpdateScore} are shown in Appendix~\ref{sec:score_function_appendix}.)
Importantly, after every update, we ensure that the subgraph induced by $S$ continues to satisfy the $\gamma$-EQC property (Line 4). Once a larger $\gamma$-EQC is found, we update $g'$ accordingly in Lines 5-6.
The process repeats until the number of missing edges in the current selected set exceeds the allowable threshold, which is recalculated based on the size of the largest $\gamma$-EQC found so far.
Finally, after all vertices in $V(g)$ are considered as seeds, we return the largest $\gamma$-EQC found (Line 10). 

\begin{algorithm}[t]
    \caption{\texttt{EQC-Heu-Pro}$(G,\gamma,g)$}\label{alg:heu-extend}
    \KwIn{A graph $G=(V,E)$, a real value $0 < \gamma < 1$, and a heuristic solution $g$}
    \KwOut{An improved $\gamma$-EQC in $G$}
    \While{true}{
        $g' \gets$ \texttt{EQC-Heu}$(G[N_G[V(g)]], \gamma, g)$\;
        \lIf{$|V(g')| = |V(g)|$}{\textbf{break}}
        $g \gets g'$\;
    }
    \Return{$g$}\;
\end{algorithm}

\noindent \underline{\textbf{Improved \texttt{EQC-Heu-Pro}.}} Although \texttt{EQC-Heu} alleviates the issue of inflexible vertex combinations, applying it directly to the entire graph is inefficient. For large graphs, maintaining set-related information incurs substantial memory overhead, and the associated random access patterns degrade cache performance, thereby reducing overall efficiency. To address these challenges, we propose \texttt{EQC-Heu-Pro}, which applies the heuristic to multiple smaller subgraphs rather than the entire graph, as in Algorithm~\ref{alg:heu-extend}.
Specifically, \texttt{EQC-Heu-Pro} first extracts the subgraph induced by the closed neighborhood $N_G[V(g)]$ of the current solution $g$ (Line 2), and then applies \texttt{EQC-Heu} to this induced subgraph to obtain a refined solution $g'$. If no improvement in size is achieved (i.e., $|V(g')| = |V(g)|$), the process terminates (Line 3); otherwise, $g$ is updated to $g'$ and the procedure is repeated (Line 4). This iterative refinement ensures efficiency and addresses the problem of inflexible solution construction without compromising solution quality.

\noindent \underline{\textbf{Time Complexity Analyses.}} 
We briefly analyze the time complexities of \texttt{EQC-Heu} and \texttt{EQC-Heu-Pro}.  
First, in Algorithm~\ref{alg:degen-opt}, computing degeneracy orderings on neighborhood subgraphs gives overall time $O(\Delta_G m)$, where $\Delta_G$ is the maximum degree. For \texttt{EQC-Heu} (Algorithm~\ref{alg:subheu}), the main cost lies in updating scores for neighbors of the current set $S$. Each update takes $O(\Delta_G s^*)$ time, and with at most $(s^*)^2$ iterations, the total time is $O(\Delta_G (s^*)^3)$. Finally, \texttt{EQC-Heu-Pro} (Algorithm~\ref{alg:heu-extend}) calls \texttt{EQC-Heu} up to $s^*$ times, yielding $O(\Delta_G (s^*)^4)$.

\section{Our Final Iterative Algorithm: \texttt{EQC-Pro}}\label{sec:EQC-pro}
Building on the insights from the bottom-up iterative framework (Section~\ref{sec:bottom-up}) and heuristic strategies (Section~\ref{sec:heu}), we introduce our final algorithm \texttt{EQC-Pro}. We first outline two key observations that guide the design of \texttt{EQC-Pro} and improve its practical performance.

\noindent \textbf{Observation 1.} Additional information can be derived from the function \texttt{solve-defect}$(k)$. In particular, when the \texttt{check} function in Algorithm~\ref{alg:bottom-up-framework} returns \texttt{true}, a solution to the \texttt{kDC} problem can be constructed directly. Moreover, since \texttt{get-k}$(\cdot)$ is a floor function, multiple values of $s$ may map to the same $k$, allowing solutions of \texttt{kDC} for one $s$ to be reused for other $s$ with the same $k$.

\noindent \textbf{Observation 2.} Within a single iteration in Algorithm~\ref{alg:bottom-up-framework}, the lower bound \( s_i \) obtained corresponds to the solution of the \texttt{kDC} problem for a particular value of \( k \). However, since our ultimate objective is to determine the lower bound \( s^* \) for the \texttt{MaxEQC} problem, there remains potential for improvement by leveraging this relationship.

Based on these observations, we (1) cache the results for previously solved values of $k$ to enable direct reuse in subsequent iterations, and (2) apply \texttt{EQC-Heu-Pro} to extend the solution obtained in each iteration, thereby potentially obtaining a larger lower bound for \texttt{MaxEQC}. Due to space limitations, the detailed pseudocode and implementation of \texttt{EQC-Pro} are provided in Appendix~\ref{sec:EQC-pro-appendix}.

\noindent \textbf{\underline{Remarks and Complexity Analyses.}}
We note that \texttt{get-k} can be directly implemented, while \texttt{solve-defect} can leverage existing \texttt{kDC} algorithms~\cite{chang2023kdefective,Chang24kDC-2,dai2024theoretically}. In our implementation, we adopt the state-of-the-art algorithm \texttt{kDC-two}~\cite{Chang24kDC-2} for \texttt{solve-defect}. 
For the time complexity of \texttt{EQC-Pro}, we note that \texttt{EQC-Pro} essentially implements the bottom-up doubling iterative framework \texttt{EQC-BU} described in Section~\ref{sec:bottom-up}. 
As discussed there, the number of iterations is \(O(\log s^*)\). 
In each iteration, the algorithm invokes \texttt{solve-defect}, which dominates the computational cost, while the accompanying heuristic method \texttt{EQC-Heu-Pro} runs in polynomial time and is invoked only once per iteration. 
The complexity of \texttt{kDC-two} (used for \texttt{solve-defect} in our implementation) is $O^*(\beta_k^n)$~\cite{Chang24kDC-2}, where $\beta_k$ is the largest real root of the equation $x^{k+2} - 2x^{k+1} + 1 = 0$, and $O^*$ omits polynomial factors. 
Note that the complexity of \texttt{solve-defect} increases monotonically with $k$. 
By our analysis of \texttt{EQC-BU}, the largest $k$ that \texttt{EQC-Pro} needs to handle is $\kappa = \texttt{get-k}(2s^*)$. 
Thus, the complexity of \texttt{EQC-Pro} is $O^*(\beta_\kappa^n)$, which is asymptotically better than the complexities of both \texttt{QClique} and \texttt{FPCE}, each of which has complexity \(O^*(2^n)\).
For space complexity, \texttt{EQC-Pro} is $O(m+n)$ as
both \texttt{kDC-two} and the heuristic components such as \texttt{EQC-Heu-Pro} maintain linear-sized structures.

\section{Experiments}\label{sec:exp}
\noindent \textbf{\underline{Algorithms.}} We compare \texttt{EQC-Pro} with state-of-the-art \textbf{exact} methods for the maximum \textbf{edge-based} quasi-clique problem (\texttt{MaxEQC}).
\begin{itemize}[leftmargin=*, topsep=0.2em]
    \item {\texttt{QClique}\footnote{We re-implemented the algorithm based on the description in \cite{ribeiro2019mqc}. Our implementation achieves better performance compared to the results reported in their paper.}:} the state-of-the-art branch-and-bound algorithm~\cite{ribeiro2019mqc}.  
    \item {\texttt{FPCE}\footnote{https://github.com/ahsanur-research/FPCE}:} a baseline adapted from the algorithm for enumerating all maximal $\gamma$-EQCs above a given size threshold~\cite{rahman2024pseudo}. We modify it for \texttt{MaxEQC} by maintaining the largest solution found so far and updating the lower bound accordingly to improve pruning.
\end{itemize}

We also evaluate the proposed bottom-up framework (Section~\ref{sec:bottom-up}) and the proposed heuristic method \texttt{EQC-Heu-Pro} (Section~\ref{sec:heu}).
\begin{itemize}[leftmargin=*, topsep=0.2em]
    \item \texttt{EQC-TD}: Algorithm~\ref{alg:top-down-framework}, the top-down iterative approach.

    \item \texttt{EQC-NH}: 
    A variant of \texttt{EQC-Pro} designed to quantify the contribution of \texttt{EQC-Heu-Pro}. All components in \texttt{EQC-Pro} that involve the use of \texttt{EQC-Heu-Pro} are removed in this variant.
\end{itemize}

All algorithms are coded in C++ and compiled using \texttt{g++} with \texttt{-O3}. Experiments are conducted on a machine with an Intel 2.60GHz CPU and 256GB RAM, running Ubuntu 20.04.6. We enforce a maximum runtime of 3 hours (10,800 seconds) per instance. 

\noindent \textbf{\underline{Datasets.}} We conduct experiments on the two collections from Network Repository, which are widely used in related studies. (1) \textbf{Real-World Collection}\footnote{http://lcs.ios.ac.cn/\~caisw/Resource/realworld\%20graphs.tar.gz}: 139 real-world graphs; and (2) \textbf{Facebook Collection}\footnote{https://networkrepository.com/socfb.php}: 114 Facebook social networks.

\begin{figure}[t]
    \centering
    \subfigure[Facebook (3-hour limit)]{
        \includegraphics[width=0.22\textwidth]{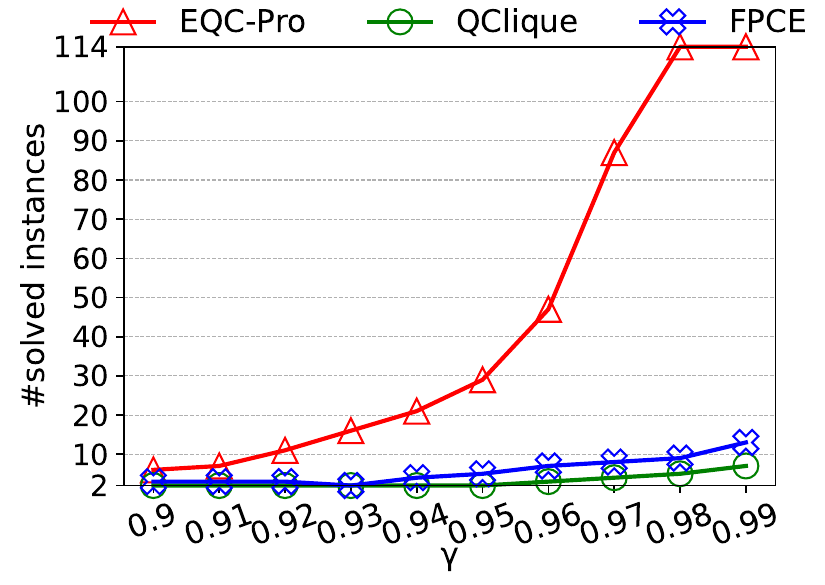}
        \label{fig:facebook3h}
    }
    \hspace{0em}
    \subfigure[Real-World (3-hour limit)]{
        \includegraphics[width=0.22\textwidth]{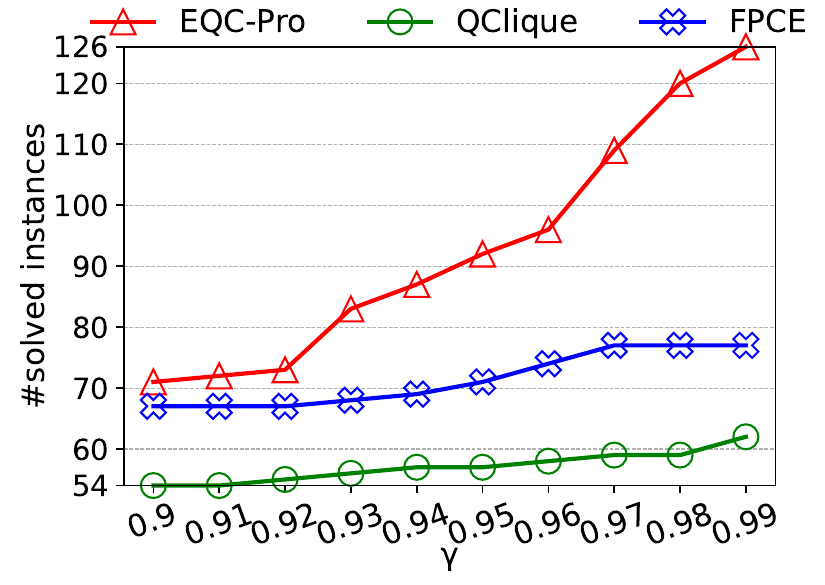}
        \label{fig:realworld3h}
    }
    
    \vspace{-1em} 
    
    \subfigure[Facebook (300-second limit)]{
        \includegraphics[width=0.22\textwidth]{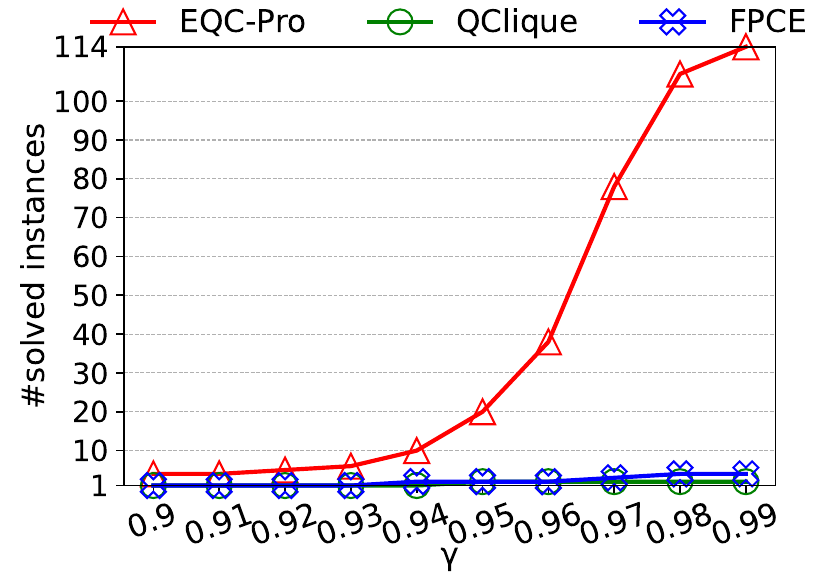}
        \label{fig:facebook300s}
    }
    \hspace{0em}
    \subfigure[Real-World (300-second limit)]{
        \includegraphics[width=0.22\textwidth]{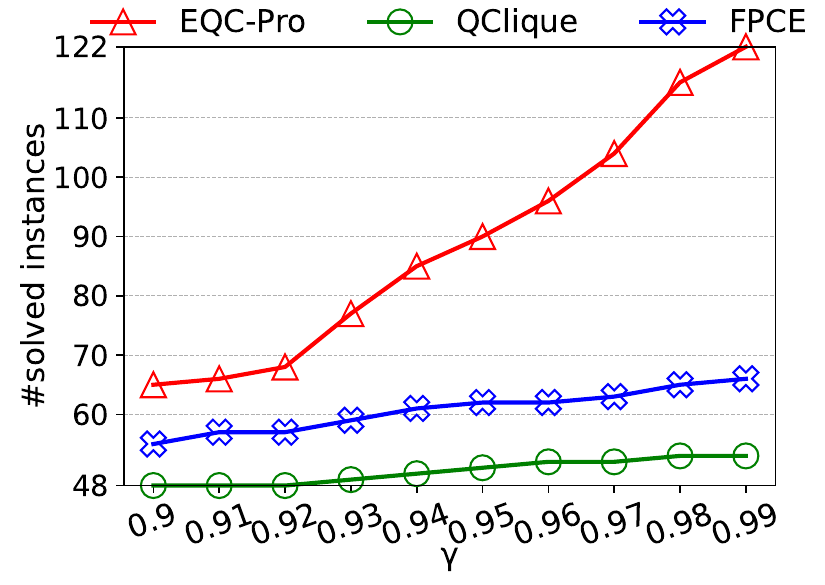}
        \label{fig:realworld300s}
    }
    \vspace{-0.2in}
    \caption{Number of solved instances with varying $\gamma$.}
    \label{Fig:3h}
\end{figure}
\begin{figure}[t]
    \centering
    \subfigure[$\gamma = 0.92$]{
        \includegraphics[width=0.22\textwidth]{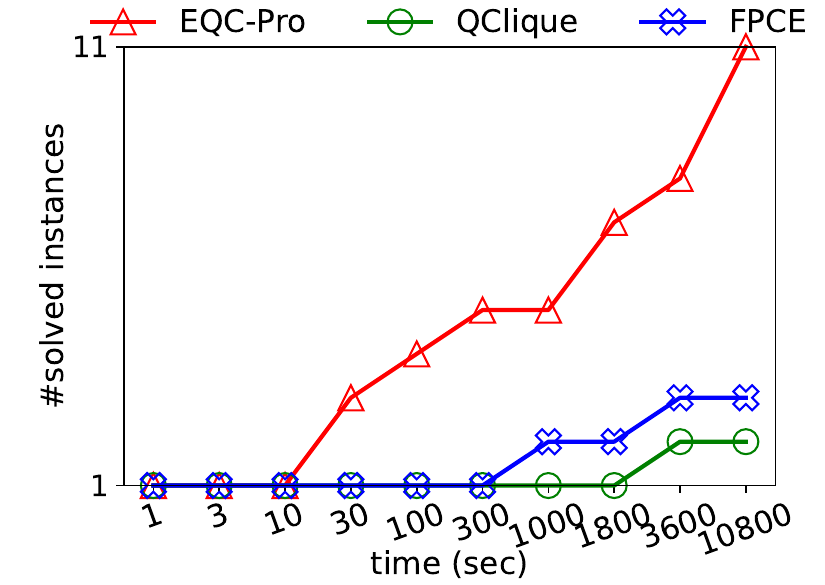}
        \label{fig:facebook092}
    }
    \subfigure[$\gamma = 0.94$]{
        \includegraphics[width=0.22\textwidth]{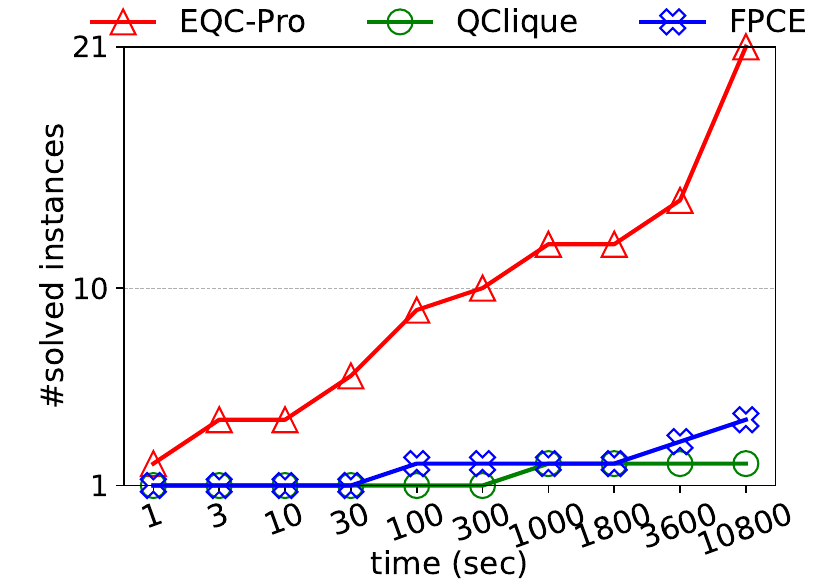}
        \label{fig:facebook094}
    }

    \vspace{-1em} 
    
    \subfigure[$\gamma = 0.96$]{
        \includegraphics[width=0.22\textwidth]{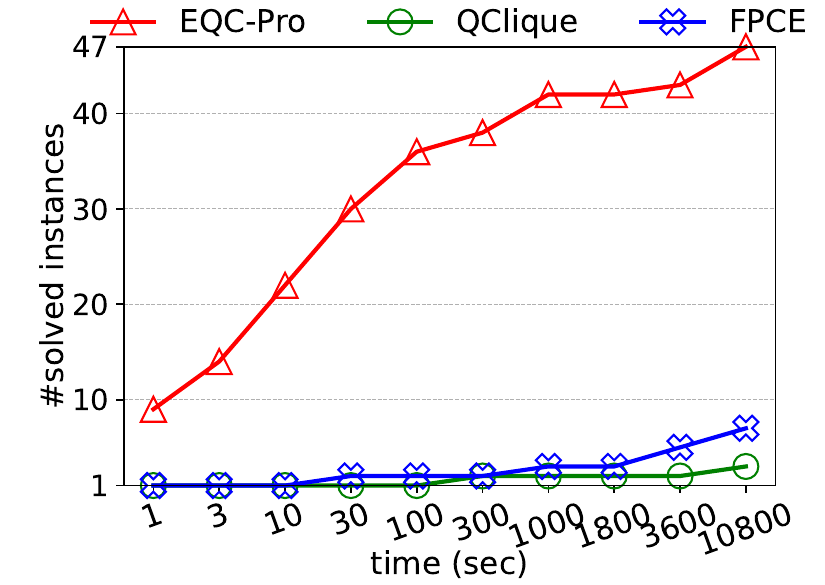}
        \label{fig:facebook096}
    }
    \subfigure[$\gamma = 0.98$]{
        \includegraphics[width=0.22\textwidth]{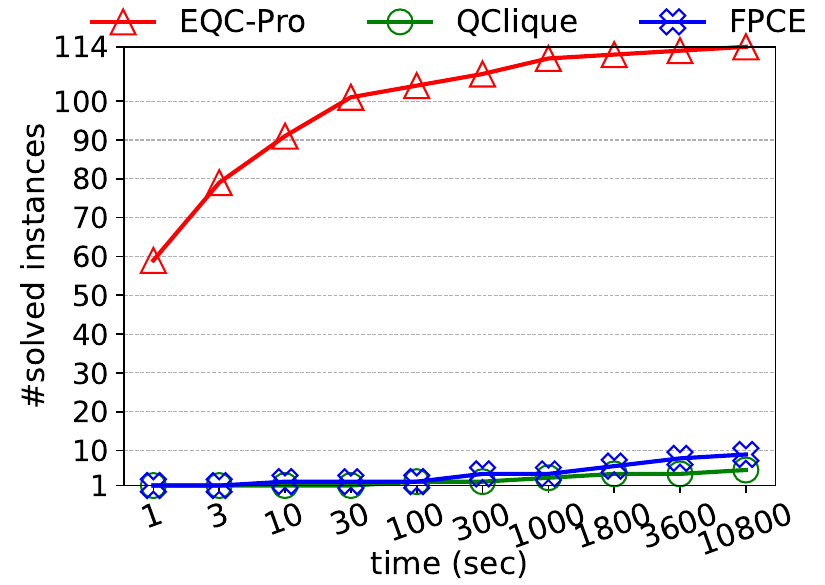}
        \label{fig:facebook098}
    }
    \vspace{-0.2in}
    \caption{Number of solved instances on Facebook.}
    \vspace{-0.17in}
    \label{Fig:facebook_various_gamma}
\end{figure}

\begin{figure}[t]
    \centering
    \subfigure[$\gamma = 0.92$]{
        \includegraphics[width=0.22\textwidth]{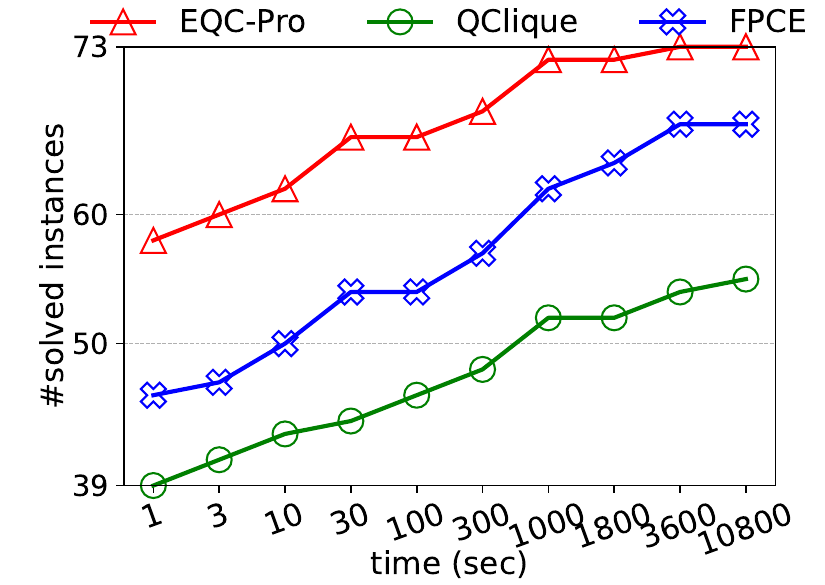}
        \label{fig:realworld092}
    }
    \subfigure[$\gamma = 0.94$]{
        \includegraphics[width=0.22\textwidth]{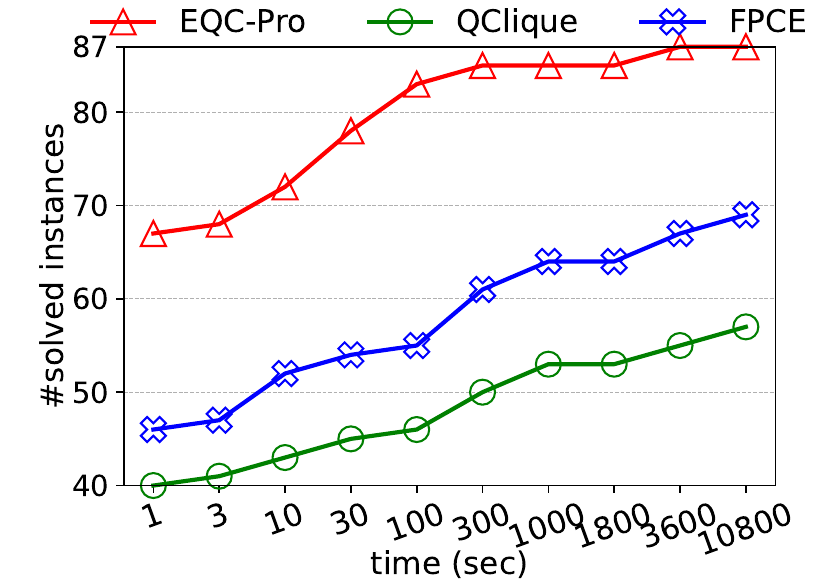}
        \label{fig:realworld094}
    }
    
\vspace{-1em} 
    
    \subfigure[$\gamma = 0.96$]{
        \includegraphics[width=0.22\textwidth]{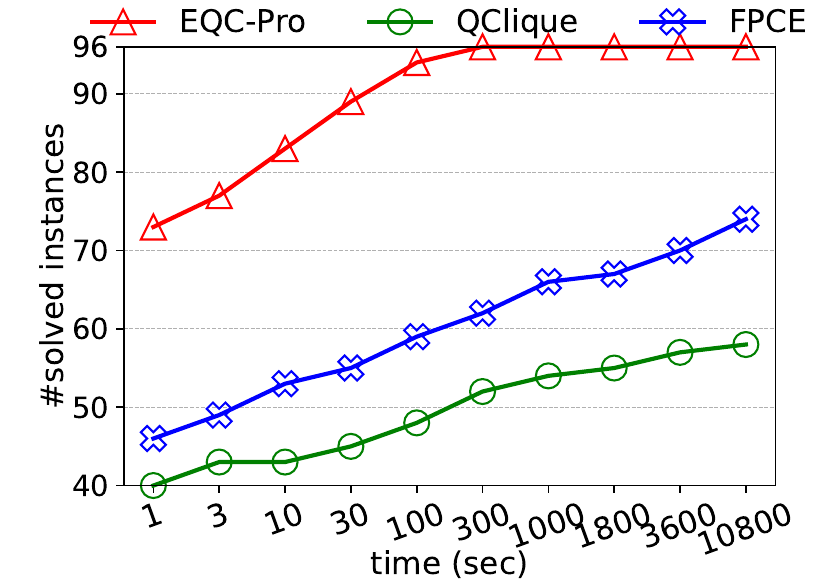}
        \label{fig:realworld096}
    }
    \subfigure[$\gamma = 0.98$]{
        \includegraphics[width=0.22\textwidth]{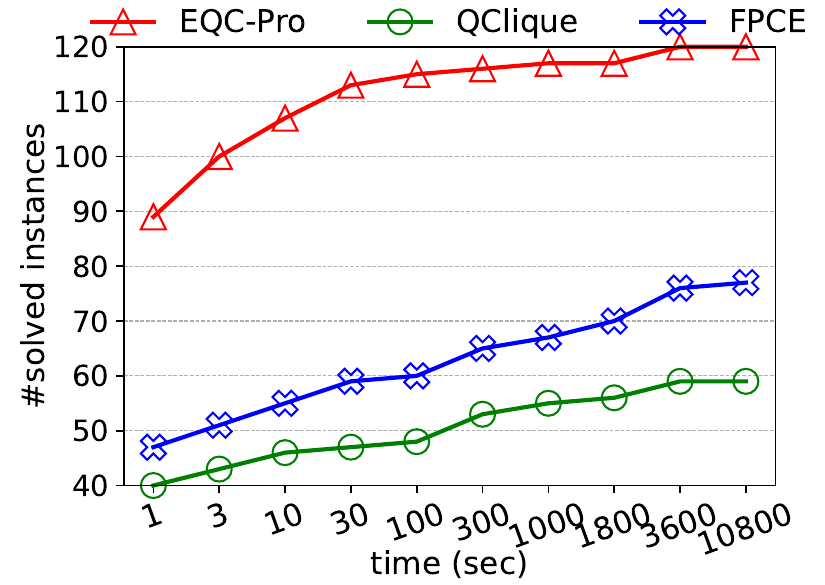}
        \label{fig:realworld098}
    }
    \vspace{-0.2in}
    \caption{Number of solved instances on Real-World.}
    \label{Fig:realworld_various_gamma}
\end{figure}

\begin{table}[t]
    \centering
    \caption{Runtime (in seconds) on 20 large graphs with $\gamma=0.95$.} 
    \label{table:instance095}
    \vspace{-0.17in}
    \scalebox{0.68}{    
    \begin{tabular}{l|lllll|ccc}
    \hline
        Graphs & \texttt{EQC-Pro} & \texttt{EQC-TD} & \texttt{EQC-NH} & \texttt{QClique} & \texttt{FPCE} & $n$ & $m$ &$s^*$ \\ \hline  
        inf-road-usa & 7.16& 7.52& \textbf{5.93} & --& 358.36& 23M& 57M & 4\\
        inf-roadNet-CA & 0.57& 0.62& \textbf{0.49} & --& 3.58& 1.9M& 5M & 4\\
        web-wikipedia2009 & \textbf{8.03} & --& 39.10& --& --& 1.8M& 9M& 34\\
        soc-pokec & \textbf{205.40} & --& 460.59& --& --& 1.6M& 44M & 35\\
        soc-lastfm & \textbf{19.66} & --& 79.49& --& --& 1.1M& 9M& 20\\
        soc-youtube-snap & \textbf{6.88} & --& 20.00& --& --& 1.1M& 5.9M& 22\\
        rt-retweet-crawl & 1.44& \textbf{1.13} & 1.26& --& 11.93& 1.1M& 4.5M& 16\\
        inf-roadNet-PA & \textbf{0.31} & 0.46& 0.39& --& 1.37& 1M& 3M& 4\\
        sc-ldoor & \textbf{211.23} & --& 212.22& --& --& 909K& 41M& 22\\
        soc-delicious & \textbf{0.67} & --& 4.64& --& --& 536K& 2.7M& 30\\
        soc-youtube & \textbf{4.31} & --& 12.79& --& --& 495K& 3M& 22\\
        sc-msdoor & \textbf{101.36} & --& 101.92& --& --& 404K& 18M& 22\\ 
        soc-twitter-follows & 0.24 & \textbf{0.16}& 0.25 & 3335.43& 5.28& 404K& 1.4M& 7\\
        ca-MathSciNet & \textbf{0.18} & --& 0.19& --& 158.14& 332K& 1.6M& 25\\
        sc-pwtk & \textbf{55.5} & --& 100.45& --& -- & 217K& 11M& 26\\
        \hline
        socfb-konect & \textbf{22.82} & 23.38& 28.99& --& 1902.99& 59M& 185M& 7\\
        socfb-uci-uni & \textbf{28.15} & 29.59& 33.39& --& --& 58M& 184M& 7\\
        socfb-A-anon & \textbf{6857.26} & --& --& --& --& 3M& 47M& 37\\
        socfb-B-anon & \textbf{1952.97} & --& --& --& --& 2.9M& 41M& 35\\
        socfb-Maine59 & \textbf{2421.99} & --& 3247.84& --& --& 9K& 486K& 41\\ \hline
    \end{tabular}}
\end{table}

\subsection{Comparison with Existing Algorithms} \label{subsection:compare_with_baseline}
\noindent \underline{\textbf{Number of Solved Instances.}} We note that an instance is considered \emph{solved} only if an exact solution is found within the given time limit. Figure~\ref{Fig:3h} presents the number of instances solved by \texttt{EQC-Pro} and two baseline methods, \texttt{QClique} and \texttt{FPCE}, under 3-hour and 300-second time limits, respectively, with varying $\gamma$ values from $0.9$ to $0.99$. We have the following observations. \textbf{First}, the number of solved instances sharply decreases as $\gamma$ becomes smaller. This is because a lower value of $\gamma$ corresponds to a looser constraint on $\gamma$-EQC, which significantly weakens the effectiveness of pruning and reduction rules in branch-and-bound methods.
\textbf{Second}, \texttt{EQC-Pro} consistently outperforms the two baselines across all tested values of $\gamma$. For example, in Figure~\ref{fig:facebook3h}, when $\gamma=0.98$ on the Facebook collection, \texttt{EQC-Pro} solves all 114 instances, while \texttt{QClique} and \texttt{FPCE} only solve 5 and 9 instances, respectively. For the Real-World collection, at \(\gamma = 0.9\) under the 300-second time limit, \texttt{QClique} and \texttt{FPCE} solve 48 and 55 instances, respectively, whereas \texttt{EQC-Pro} solves 65 instances as shown in Figure~\ref{fig:realworld300s}.
The significant speedup achieved by \texttt{EQC-Pro} results from transforming the problem into a sequence of hereditary subproblems, which substantially accelerates the search process. Compared to the Real-World collection, the Facebook graphs are denser, making the problem more challenging to solve under the same values of $\gamma$. The advantage of \texttt{EQC-Pro} becomes more obvious at higher $\gamma$, where the problem is moderately constrained. However, as $\gamma$ decreases, the search space increases dramatically, resulting in greater computational difficulty and reducing the performance gap between \texttt{EQC-Pro} and the baselines.

Figures~\ref{Fig:facebook_various_gamma} and~\ref{Fig:realworld_various_gamma} illustrate the number of solved instances over time on the Facebook and Real-World collections, respectively, for \(\gamma \in \{0.92, 0.94, 0.96, 0.98\}\). As shown, both baseline methods consistently exhibit poor performance on these datasets regardless of the value of \(\gamma\). On the Facebook collection, when \(\gamma = 0.98\), \texttt{EQC-Pro} successfully solves all 114 instances, whereas the best of the two baselines solves only 9. Moreover, across all tested values of \(\gamma\), \texttt{EQC-Pro} solves more instances within 30 seconds than either baseline does within 3 hours. For example, at \(\gamma = 0.96\), \texttt{EQC-Pro} solves 30 instances in only 30 seconds, compared to only 3 and 7 solved by \texttt{QClique} and \texttt{FPCE}, respectively, even with a 3-hour limit. On the Real-World collection, \texttt{EQC-Pro} also consistently outperforms the two baselines under all tested values of \(\gamma\). Specifically, at \(\gamma = 0.94\), \texttt{EQC-Pro} solves 77 instances within only 30 seconds, while the best baseline solves only 68 instances even with a 3-hour limit.

\noindent \underline{{\textbf{Evaluation on Large Graphs.}}} Table~\ref{table:instance095} reports the performance of the baselines \texttt{QClique} and \texttt{FPCE}, as well as our \texttt{EQC-Pro} (and its variants), on large graphs with $\gamma = 0.95$. Specifically, we select 15 largest graphs from the Real-World collection and 5 largest graphs from the Facebook collection, excluding extremely difficult cases in which all three algorithms fail to solve. The instances are then sorted in descending order of graph size in Table~\ref{table:instance095}. We use ``--'' to indicate a timeout (i.e., 3 hours) and highlight the best results in \textbf{bold}.
It can be observed from Table~\ref{table:instance095} that \texttt{EQC-Pro} consistently outperforms the two baselines, \texttt{QClique} and \texttt{FPCE}, across all instances. Specifically, in 13 out of 20 cases, both baselines fail to produce a solution within 3 hours, whereas \texttt{EQC-Pro} succeeds. For example, on the \texttt{soc-delicious} instance, \texttt{EQC-Pro} completes in 0.67 seconds, while both baselines time out, implying a speedup of at least $1.6\times10^4$. Furthermore, when considering the better of the two baselines, \texttt{EQC-Pro} achieves a speedup of over $1000\times$ in 4 instances, at least $100\times$ in 9 instances, and over $20\times$ in 14 instances.

\noindent \underline{\textbf{Scalability Test.}}
We conduct scalability experiments to evaluate the performance of \texttt{EQC-Pro}. The results demonstrate that \texttt{EQC-Pro} significantly outperforms the baseline algorithms. Detailed experimental settings and results are reported in Appendix~\ref{sec:scale-appendix}.

\subsection{Comparison with Our Variants}
\noindent \underline{\textbf{Evaluation on Large Graphs.}} Table~\ref{table:instance095} reports the performance of \texttt{EQC-TD} and \texttt{EQC-NH} under $\gamma=0.95$. 
We first consider the effect of different iterative frameworks by comparing \texttt{EQC-Pro} and \texttt{EQC-TD}. From Table~\ref{table:instance095}, we know that \texttt{EQC-Pro}, benefiting from its bottom-up doubling iterative framework, can solve 13 additional instances that time out under the top-down iterative method \texttt{EQC-TD}. Specifically, on \texttt{ca-MathSciNet}, it achieves a speedup of at least $6\times10^4$. For the remaining 7 instances where both \texttt{EQC-Pro} and \texttt{EQC-TD} complete within the time limit, the performance of \texttt{EQC-TD} is comparable to that of \texttt{EQC-Pro}. This is mainly due to the small final solution size $s^*$ in these cases. For instance, in \texttt{soc-twitter-follows} and \texttt{socfb-uci-uni}, the maximum size of the 0.95-EQC is only 7, resulting in a limited allowance for missing edges. This tight constraint indirectly improves the accuracy of upper bound estimations, narrowing the performance gap between \texttt{EQC-Pro} and \texttt{EQC-TD}. 

We next consider \texttt{EQC-Heu-Pro} (Algorithm~\ref{alg:heu-extend}) by comparing \texttt{EQC-Pro} and \texttt{EQC-NH}.
The integration of \texttt{EQC-Heu-Pro} leads to a reduction in runtime. Among the 20 instances in Table~\ref{table:instance095}, 17 instances exhibit accelerated performance, with two additional time-out cases successfully solved by \texttt{EQC-Pro}. 
In most cases, incorporating \texttt{EQC-Heu-Pro} yields a clear efficiency gain. We now examine the few cases with less favorable outcomes. In \texttt{inf-road-usa}, for instance, \texttt{EQC-NH} outperforms \texttt{EQC-Pro}, consistent with our earlier observation that the heuristic component \texttt{Degen-Opt}, retained in \texttt{EQC-NH}, can already produce high-quality solutions in certain settings. The speedup is also marginal in \texttt{socfb-uci-uni} and \texttt{socfb-konect}, where \texttt{EQC-Pro} is faster by only about 5 seconds, despite total running times of approximately $28.15$ and $22.82$ seconds, respectively. This limited improvement mainly stems from the dominance of \texttt{Degen-Opt} in the total computation time on large-scale graphs. As both \texttt{EQC-Pro} and \texttt{EQC-NH} retain this component, the observed speedup is less obvious, despite a substantial reduction in the number of required iterations.

\section{Related Work}\label{sec:related-work}
The concept of edge-based \(\gamma\)-quasi-clique ($\gamma$-EQC) was defined in~\cite{abello1998massive,abello2002massive}, and has also been referred to as dense subgraph~\cite{Long2010,Tsourakakis2013}, near-clique~\cite{Brakerski2011,Tadaka2016,komusiewicz2015connected}, or pseudo clique~\cite{rahman2024pseudo}. Existing exact approaches for solving \texttt{MaxEQC} fall into two main categories: mixed integer programming (MIP) and branch-and-bound techniques. Since MIP-based methods~\cite{pattillo2013maximum,marinelli2020lp,veremyev2016exact} are difficult to scale in practice, we focus on branch-and-bound algorithms. The first branch-and-bound algorithm for this problem was proposed by Pajouh et al.~\cite{pajouh2014branch}, who introduced a novel upper bounding technique and demonstrated strong performance. Subsequently, \texttt{QClique}~\cite{ribeiro2019mqc} incorporated a new upper bound and showed that their approach is competitive with both MIP techniques~\cite{pattillo2013maximum,veremyev2016exact} and other branch-and-bound methods~\cite{pajouh2014branch}. Although \texttt{QClique} achieves competitive performance, its branching scheme is identical to that of the maximal $\gamma$-EQC enumeration method \texttt{PCE}~\cite{Uno2010}, and its pruning strategy relies on a computationally expensive upper bound. \texttt{FPCE}~\cite{rahman2024pseudo} also builds on \texttt{PCE} and specifically targets the maximal $\gamma$-EQC enumeration problem, introducing several refined bounding techniques and achieving state-of-the-art performance in both maximal $\gamma$-EQC enumeration and \texttt{MaxEQC}.
Given that \texttt{MaxEQC} is NP-hard~\cite{pattillo2013maximum}, various heuristics have been proposed for approximate solutions. Early approaches include GRASP-based methods and local search~\cite{abello2002massive,brunato2007quasi}. Later studies explored greedy strategies~\cite{Tsourakakis2013} and, recently, advanced vertex selection and iterative refinement techniques~\cite{wang2021nuqclq,liu2024optimization}.

As mentioned earlier, there is a quasi-clique defined based on vertex degree: a degree-based $\gamma$-quasi-clique requires each vertex in the subgraph connects to at least a \(\gamma\)-proportion of the remaining vertices~\cite{MIH99,sanei-mehri2021largest}. \texttt{DDA}~\cite{Pastukhov2018gamma} introduced a hybrid strategy that combines MIP formulations with an iterative refinement process. More recently, Xia et al.~\cite{xia2025iterqc} proposed \texttt{IterQC}, which reformulates the original problem as a series of simpler dense subgraph problems. Another closely related problem is the corresponding maximal enumeration problem, for which several branch-and-bound methods, including \texttt{Quick}~\cite{liu2008pruning} and \texttt{Quick+}~\cite{khalil2022parallel}, have been developed. Most recently, Yu and Long~\cite{Yu2023} proposed the state-of-the-art method \texttt{FastQC} by introducing new branching and pruning strategies. 

For comprehensive surveys on cohesive subgraphs, see~\cite{Chang2018cosub, Fang2020cosub, Fang2021cosub, Huang2019cosub, Lee2010cosub}.

\section{Conclusion}\label{sec:conclusion}
In this paper, we proposed an efficient iterative algorithm \texttt{EQC-Pro} for the \texttt{MaxEQC} problem. Reducing \texttt{MaxEQC} to hereditary subproblems improves theoretical efficiency, while a bottom-up doubling framework and novel heuristics further boost performance. Experiments show that \texttt{EQC-Pro} outperforms state-of-the-art methods. Future work includes extending \texttt{EQC-Pro} to weighted graphs.

\begin{acks}
This research is partially supported by the Guangdong Basic and Applied Basic Research Foundation (No. 2025A1515060015), by the Shenzhen Science and Technology Program (No. GXWD20231129111306002 and KJZD20231023094701003), by the Major Key Project of Pengcheng Laboratory (No. PCL2024A05), and by the Key Laboratory of Interdisciplinary Research of Computation and Economics (Shanghai University of Finance and Economics), Ministry of Education.
\end{acks}

\clearpage
\bibliographystyle{ACM-Reference-Format}
\bibliography{ref}


\appendix

\section{Omitted Details of Top-down Iterative Framework in Section~\ref{sec:top-down}}\label{sec:top-down-appendix}
\subsection{Our Top-down Iterative Baseline \texttt{EQC-TD}}
We first detail our top-down iterative baseline algorithm \texttt{EQC-TD} (mentioned in Section~\ref{sec:top-down}) and demonstrate its correctness.

\noindent \textbf{\underline{Overview.}} 
We observe that the maximum $\gamma$-EQC $g^*$ is also a maximum $k^*$-defective clique in $G$, where $k^*=\texttt{get-k}(s^*)$, which follows directly from the definitions. 
Motivated by this equivalence, we design the top-down iterative framework that reduces \texttt{MaxEQC} to a sequence of \texttt{kDC} subproblems. 
Since the optimal size $s^*$ is unknown, the framework starts from an initial upper bound and iteratively tightens it by solving the corresponding \texttt{kDC} problem. 
At each iteration, the obtained solution either certifies feasibility or yields a strictly smaller upper bound on $s^*$, thereby narrowing the search space until the optimal $\gamma$-EQC is identified. 
We then detail the algorithm and provide a formal proof of its correctness.

\begin{algorithm}[t]
    \caption{A Top-down Iterative Framework: \texttt{EQC-TD}}\label{alg:top-down-framework}
    \KwIn{A graph $G=(V,E)$ and a real value $0 < \gamma <  1$}
    \KwOut{The size of the maximum $\gamma$-EQC in $G$}  
    $s_0 \gets \texttt{UpperBound}(G)$; $k_1 \gets \texttt{get-k}(s_0)$; $i\gets 1$\;
    \While{true}{
        $s_i \gets$ \texttt{solve-defect}($k_i$)\;
        \If{$k_i = \emph{\texttt{get-k}}(s_i)$}{
            $s^* \gets s_i$; \Return{$s^*$}\;
        }
        $k_{i+1} \gets \texttt{get-k}(s_i)$; $i \gets i+1$\;
    }
\end{algorithm}

\noindent \textbf{\underline{\texttt{EQC-TD} Algorithm.}}  Our top-down iterative framework (named as \texttt{EQC-TD}) is shown in Algorithm~\ref{alg:top-down-framework}. Line 1 initializes an upper bound as \(s_0 \) via \texttt{UpperBound} (described later). This upper bound is iteratively refined within the while-loop (Lines 2-6). In each iteration, \texttt{EQC-TD} assumes the existence of a size‑\(s_{i-1}\) \(\gamma\)-EQC and computes the corresponding upper bound on the number of missing edges, denoted as \(k_i = \texttt{get-k}(s_{i-1})\). Using this value, \texttt{EQC-TD} invokes \texttt{solve-defect} to obtain the maximum size \(s_i\) of a \(k_i\)-defective clique. We show that this procedure provides a tighter upper bound for the optimal solution \(s^*\). The algorithm terminates when the \(k\)-sequence converges, i.e., when \(k_i = \texttt{get-k}(s_i)\) in Line 4, and returns the corresponding size \(s_i\) as the final result in Line 5.

\begin{algorithm}[t]
    \caption{\texttt{UpperBound}$(G)$}\label{alg:upperbound}
    \KwOut{An upper bound $ub$ of $s^*$}
    $ub \gets (1 + \sqrt{1 + 8m/\gamma}) / 2$; $k \gets \texttt{get-k}(ub)$\;
    \While{true}{
        $ub \gets$ \texttt{ComputeUB}$(k)$\;
        \lIf{$\texttt{get-k}(ub) \geq k$}{\textbf{break}}
        $k \gets \texttt{get-k}(ub)$\;
    }
    \Return{$ub$}\;
\end{algorithm}

\noindent \textbf{\underline{Correctness.}}
We next prove the correctness of our top-down iterative framework \texttt{EQC-TD} for solving \texttt{MaxEQC} in Algorithm~\ref{alg:top-down-framework}. We first establish two supporting lemmas.

    \begin{lemma}\label{lem:exit_lb}
        Given a real number $0<\gamma<1$, recall that \(s^*\) denotes the size of the maximum edge-based \(\gamma\)-quasi-clique. If a value \(s\) satisfies the condition $s = \texttt{solve-defect}(\texttt{get-k}(s))$, then we have \(s \leq s^*\).
    \end{lemma}
    \begin{proof}
        Let 
        $k =\texttt{get-k}(s) =\left\lfloor \binom{s}{2} \times (1-\gamma)\right\rfloor$.
        Since we have \(\texttt{solve-defect}(k)=s\), there exists a subgraph $g$ of $G$ with \(|V(g)|=s\) that is a \(k\)-defective clique. Thus, we have
        $
        |E(g)| \;\ge\; \binom{s}{2} \;-\; k
        \;\ge\;
        \binom{s}{2} \;-\;\binom{s}{2}(1-\gamma)
        =
        \binom{s}{2}\,\gamma$,
        which implies that \(g\) is an edge-based \(\gamma\)-quasi‑clique of size \(s\). By the definition of \(s^*\) as the maximum size of any edge-based \(\gamma\)-quasi‑clique, we conclude that \(s \le s^* \).
    \end{proof}

    \begin{lemma}\label{lem:strict_decrease}
        Let \( s^* \) be the optimal size. Let \( ub>s^*\) be an upper bound before an iteration (i.e., an execution of the while-loop in Lines 2-6 of Algorithm~\ref{alg:bottom-up-framework}), where $ub$ may be the $s$ value obtained in the last iteration. Define $k = \texttt{get-k}(ub)$ and $ub'$ as the $s$ value after the iteration, i.e., $ub' = \texttt{solve-defect}(k)$. Then, we have $s^* \le ub' < ub$. In other words, \(ub'\) is a strictly tighter upper bound on \(s^*\).
    \end{lemma}
\begin{proof}
    Since \texttt{get-k} and \texttt{solve-defect} are non-decreasing, by Proposition~\ref{pro:bin-property}, we know $\texttt{solve-defect}(\texttt{get-k}(s^*)) \ge s^*$. 
    Applying this to \(ub > s^*\), we obtain $ub' = \texttt{solve-defect}(\texttt{get-k}(ub)) \ge \texttt{solve-defect}(\texttt{get-k}(s^*))\ge s^*$
    On the other hand, since $ub > s^*$ and with Proposition~\ref{pro:bin-property}, we have $ub' = \texttt{solve-defect}(\texttt{get-k}(ub)) < ub$.
Combining the two bounds yields $s^* \le ub' < ub$.
\end{proof}

\begin{proposition}
    Algorithm~\ref{alg:top-down-framework} can correctly computes the size of the maximum edge-based $\gamma$-quasi-clique in finite steps.
\end{proposition}
\begin{proof}
    By Lemma~\ref{lem:exit_lb}, whenever the termination condition of the while-loop (Line 4) is satisfied, the returned value satisfies \(s_i \le s^*\). By Lemma~\ref{lem:strict_decrease}, as long as the termination condition is not met, the sequence \(\{s_0, s_1, \ldots, s_i\}\) is strictly decreasing and remains an upper bound on \(s^*\). Further, the $s$ value at each iteration satisfies \(s_i \ge s^*\) until the termination condition is triggered. Thus, the only possibility at termination is that \(s_i = s^*\) and the algorithm can be terminated in finite steps.
\end{proof}
\subsection{Upper Bound and Correctness}
Note that the proof in the previous subsection reveals that \texttt{EQC-TD} essentially iteratively computes tighter upper bounds. Thus, it is possible to compute an initial upper bound before starting the iteration. In this subsection, we introduce the algorithm for solving the upper bound in \texttt{EQC-TD} and provide a proof of its correctness.

\noindent \textbf{\underline{Upper Bound.}} 
An upper bound is commonly used in the literature for \texttt{MaxEQC}~\cite{veremyev2016exact, pattillo2013maximum}, which is used to estimate the size of the maximum \(\gamma\)-EQC by solving \(m \geq \gamma \cdot \binom{ub}{2}\), yielding \(ub \leq (1 + \sqrt{1 + 8m/\gamma}) / 2\). However, this bound is often overly coarse. To obtain a tighter upper bound, we incorporate the coloring-based upper bounding strategy from the \texttt{kDC} problem~\cite{chang2023kdefective,dai2024theoretically} into \emph{an iterative framework tailored for the $\gamma$-EQC setting}. 
Our method for computing upper bounds for \texttt{MaxEQC} is based on a key observation. In our top-down iterative framework (Algorithm~\ref{alg:top-down-framework}), we maintain upper bounds for the optimal values \(k^* = \texttt{get-k}(s^*)\) and \(s^*\) at each iteration. By consistently choosing \(k\) and \(s\) as upper bounds, each iteration produces progressively tighter upper bounds. Thus, \emph{if the \(k\) value in some iteration is an upper bound of \(k^*\), then the upper bound computed for the corresponding \texttt{kDC} problem is also an upper bound of \(s^*\)}. 
Based on this, we summarize our upper bound method \texttt{UpperBound} in Algorithm~\ref{alg:upperbound}, where \texttt{ComputeUB} is an upper bound estimation procedure for \texttt{kDC}~\cite{chang2023kdefective,dai2024theoretically}. 
Specifically, \texttt{UpperBound} first initializes a trivial $ub$ as $(1 + \sqrt{1 + 8m/\gamma}) / 2$ and sets the corresponding value of $k$ in Line 1. Then, \texttt{UpperBound} iteratively invokes \texttt{ComputeUB} with different values of $k$ (Lines 2-6 of Algorithm~\ref{alg:upperbound}) until the condition $\texttt{get-k}(ub) \geq k$ is satisfied in Line 4. Finally, we return $ub$ as the upper bound of $s^*$ in Line 6.

We first present the existing \texttt{ComputeUB} algorithm~\cite{chang2023kdefective,dai2024theoretically} for computing an upper bound for \texttt{kDC}, which is used in our upper bound computation algorithm \texttt{UpperBound} (Algorithm~\ref{alg:upperbound}).
We then prove that in our \texttt{UpperBound}, the iterative procedure (Lines 2-5 of Algorithm~\ref{alg:upperbound}) produces a strictly decreasing sequence of upper bounds, each of which remains a valid upper bound of the optimal value~$s^*$, i.e., the size of the maximum $\gamma$-EQC. This result can also be used to demonstrates the correctness of our \texttt{UpperBound}.

 \begin{algorithm}[t]    \caption{\texttt{ComputeUB}$(k)$~\cite{chang2023kdefective,dai2024theoretically}}\label{alg:computeub}
    \KwOut{An upper bound of the size of the maximum $k$-defective clique.}
    $\text{colorCount[]}\gets \texttt{GreedyColor}(G)$\;
    Let $\text{numColors}$ be the total number of colors used\;
    $\text{maxCount} \gets 0$\;
    \ForEach{$c = 1$ \emph{to} $\text{numColors}$}{
        $\text{maxCount} \gets \max(\text{maxCount}, \text{colorCount}[c])$\;
    }

    Perform a binary search over $l \in [0, \text{maxCount}]$ to find the largest $l$ such that $\sum_{c=1}^{\text{numColors}} \binom{m_c}{2} \leq k, \text{ where } m_c = \min(l, \text{colorCount}[c])$\;

    $ub \gets 0$\;
    \ForEach{$c = 1$ \emph{to} $\text{numColors}$}{
        $m_c \gets \min(l, \text{colorCount}[c])$\;
        $k \gets k - \binom{m_c}{2}$\;
        $ub \gets ub + m_c$\;
    }
    $ub \gets ub + \lfloor k / (l+1) \rfloor$\;
    \Return{$ub$}\;

    \SetKwFunction{FGreedyColor}{GreedyColor}
    \SetKwProg{Fn}{Function}{:}{}
    \Fn{\FGreedyColor{$G$}}{
        \ForEach{vertex $u$ in reverse degeneracy order}{
            Assign to $u$ the smallest color that has not been assigned to any of its colored neighbors\;
        }
        \Return an array containing the number of vertices assigned to each color\;
    }
\end{algorithm}

\noindent \textbf{\underline{\texttt{ComputeUB} Algorithm.}} The key idea of \texttt{ComputeUB} is to apply a vertex coloring strategy that assigns different colors to adjacent vertices, partitioning the graph into independent sets (i.e., each color class consists non-adjacent vertices). \texttt{ComputeUB} focuses only on missing edges within each color class. During the selection process, vertices are greedily added such that each new vertex introduces the fewest missing edges among already selected vertices of the same color. This process continues until the total number of missing edges reaches or exceeds the threshold $k$. Since missing edges between different color classes are ignored, the number of selected vertices provides an upper bound on the size of the maximum $k$-defective clique. We remark that since the greedy selection is based on a fixed coloring order rather than $k$, \texttt{ComputeUB} is monotonic in $k$: for $k' > k$, $\texttt{ComputeUB}(k') \ge \texttt{ComputeUB}(k)$.

\texttt{ComputeUB} is shown in Algorithm~\ref{alg:computeub}. The graph is first greedily colored in Line 1. Then, vertices are selected greedily to minimize new missing edges within each color class. To improve the efficiency, we use a batch-wise selection strategy: initially, up to $l$ vertices per color class are selected, followed by at most one more vertex per class. Since the coloring order is fixed, $l$ can be efficiently found via binary search (Lines 2-6), after which any remaining selections are made iteratively (Lines 7--12). Note that as long as the coloring order is fixed, the correctness is guaranteed. In our implementation, we use the reverse degeneracy ordering for coloring.

\noindent \textbf{\underline{Correctness.}} 
We first note that \texttt{ComputeUB} (Algorithm~\ref{alg:computeub}) is monotone, i.e., for any $k_1 \geq k_2$, we have $\texttt{ComputeUB}(k_1) \geq \texttt{ComputeUB}(k_2)$.
Based on this monotonicity, we present the following lemma.
\begin{lemma} \label{lem:ub-decrease}
During the iterative process (Lines 2-5) of Algorithm~\ref{alg:upperbound}, the computed upper bound $ub$ forms a non-increasing sequence and eventually terminates at a value that is an upper bound of $s^*$, i.e., the size of the maximum $\gamma$-EQC.
\end{lemma}
\begin{proof}
Algorithm~\ref{alg:upperbound} iteratively computes a sequence of upper bounds on optimal value $s^*$. We prove that upon termination, the final upper bound is a valid upper bound on $s^*$, and that the algorithm terminates in finitely many steps.
Initially, in Line 1 of Algorithm~\ref{alg:upperbound}, the value of $k$ is set by assuming that the optimal value is $ub$. Specifically, the algorithm computes the maximum number of missing edges allowed in a $\gamma$-EQC of size $ub$, which serves as an upper bound on the number of missing edges in the corresponding maximum $k^*$-defective clique, i.e., $k^* = \texttt{get-k}(s^*)$.

Consider two consecutive iterations. Let $ub_1$ denote the current upper bound of the optimal value $s^*$ at the beginning of an iteration, and let $k_1$ be the corresponding upper bound on the number of missing edges, i.e., $k_1 \ge k^*$. Let $k_2 = \texttt{get-k}(ub_1)$ be the maximum number of missing edges associated with a $\gamma$-EQC of size $ub_1$, and let $ub_2 = \texttt{ComputeUB}(k_2)$ be the new upper bound computed for the next iteration. In each subsequent iteration, the algorithm updates $ub_1 \leftarrow ub_2$ and $k_1 \leftarrow k_2$.
Observe that the while loop proceeds only if $k_2=\texttt{get-k($ub_1$)} < k_1$, ensuring that the value of $k$ strictly decreases with each iteration. Since the function \texttt{ComputeUB} is monotonic with respect to $k$, 
\(
ub_2 = \texttt{ComputeUB}(k_2) \le \texttt{ComputeUB}(k_1) = ub_1,
\)
and thus, the sequence of upper bounds is non-increasing.

Moreover, \texttt{ComputeUB} returns a valid upper bound on the size of the maximum $k$-defective clique for the given $k$. By construction, $k_2=\texttt{get-k}(ub_1) \ge \texttt{get-k}(s^*) = k^*$, and thus $ub_2=\texttt{ComputeUB}(k_2) \ge  \texttt{ComputeUB}(k^*)= s^*$ throughout the process.
Since $k$ is a non-negative integer that can only take finitely many values bounded below by $k^*$, and since $k$ strictly decreases in each iteration, the algorithm must terminate after finitely many steps. Upon termination, the final $ub$ remains a valid upper bound on $s^*$.
\end{proof}

\begin{algorithm}[t]
    \caption{Auxiliary Subroutines of \texttt{EQC-Heu}}\label{alg:heu-submodules}

    \SetKwProg{Fn}{Function}{:}{}

    \SetKwFunction{FAddBestVertex}{AddBestVertex}
    \Fn{\FAddBestVertex{$S,C$}}{
        $P \gets \{u \in C \mid d_G^S(u) = \max_{v \in C} d_G^S(v)\}$\;
        $P' \gets \{u \in P \mid \text{score}(u) = \max_{v \in P} \text{score}(v)\}$\;
        $u \gets$ a randomly chosen vertex from $P'$\;
        move $u$ from $C$ to $S$\;
        score[$u$] $\gets 0$\;
    }

    \SetKwFunction{FRemoveWorstVertex}{RemoveWorstVertex}
    \Fn{\FRemoveWorstVertex{$S,C$}}{
        $P \gets \{u \in S \mid d_G^S(u) = \min_{v \in S} d_G^S(v)\}$\;
        $P' \gets \{u \in P \mid \text{score}(u) = \min_{v \in P} \text{score}(v)\}$\;
        $u \gets$ a randomly chosen vertex from $P'$\;
        move $u$ from $S$ to $C$\;
        score[$u$] $\gets 0$\;
    }

    \SetKwFunction{FUpdateScore}{UpdateScore}
    \Fn{\FUpdateScore{$S,C$}}{
        \ForEach{vertex $v \in S$}{
            score[$v$] $\gets$ score[$v$] $+ d_G^C(v) - \Delta_G$\;
        }
        \ForEach{vertex $v \in C$}{
            score[$v$] $\gets$ score[$v$] $+ d_G^S(v)$\;
        }
    }
\end{algorithm}

\section{Score Functions in \texttt{EQC-Heu} (Algorithm~\ref{alg:subheu})}\label{sec:score_function_appendix} 
For completeness, we present the auxiliary subroutines used in \texttt{EQC-Heu} (Algorithm~\ref{alg:subheu} in Section~\ref{sec:heu}). We note that our \texttt{EQC-Heu} uses the score function from~\cite{wang2021nuqclq}. The auxiliary subroutines are summarized in Algorithm~\ref{alg:heu-submodules}. 
The procedure \texttt{AddBestVertex} (Lines 1–6) selects from the candidate set $C$ the vertex with the largest degree towards $S$, breaking ties by the score function, and then moves it into $S$ with its score reset. 
Symmetrically, \texttt{RemoveWorstVertex} (Lines 7–12) identifies the vertex in $S$ with the smallest internal degree, again breaking ties by score, and moves it back to $C$ with its score reset. 
Finally, \texttt{UpdateScore} (Lines 13–17) updates the scores of all vertices: vertices in $S$ are adjusted according to their connections to $C$ relative to $\Delta_G$, while vertices in $C$ are updated based on their degree towards $S$. 
Together, these subroutines support the main heuristic process by guiding vertex movements and maintaining score information during the iterations.

\section{Omitted Details of \texttt{EQC-Pro} in Section~\ref{sec:EQC-pro}} \label{sec:EQC-pro-appendix}
\begin{algorithm}[t]
  \caption{Our Final Algorithm: \texttt{EQC-Pro}}
  \label{alg:improved-framework-memo}
\KwIn{A graph $G=(V,E)$ and a real value $0 < \gamma < 1$}
\KwOut{The size of the maximum $\gamma$-EQC in $G$}
  \textbf{global} \textit{memo} \textbf{as} map from $k$ to $s'$\;
  $g \gets \texttt{Degen-Opt}(G, \gamma)$\; 
  $g \gets \texttt{EQC-Heu-Pro}(G, g, \gamma)$\;
  $gap \gets 1$; $lb,lb_{new} \gets |V(g)|$;    $ub \gets |V|$\;

    \While{$lb+gap\le ub$}{
    $s \gets lb + gap$; $k \gets \texttt{get-k}(s)$\;
    $s' \gets \texttt{MemoizedSearch}(k)$\;
    \lIf{$s' < s$}{\textbf{break}}
    $lb_{new} \gets s'$\;
    $gap \gets 2 \times gap$\;
  }
  $ub_{\text{new}} \gets  lb + gap - 1$\;

  \While{$lb_{\text{new}} \neq ub_{\text{new}}$}{
    $mid \gets \left\lceil (lb_{\text{new}} + ub_{\text{new}})/2 \right\rceil$; $k \gets \texttt{get-k}(mid)$\;
    $s' \gets \texttt{MemoizedSearch}(k)$\;
    \eIf{$s' \ge mid$}{
      $lb_{\text{new}}\gets mid $\;
    }
    {
      $ub_{\text{new}} \gets mid - 1$\;
    }
  }
  \Return{$lb_{\text{new}}$}\;
  \SetKwFunction{FMain}{MemoizedSearch}
  \SetKwProg{Fn}{Function}{:}{}
  \Fn{\FMain{$k$}}{
    \eIf{$k \in \textit{memo}$}{ \tcp{memoization}
      $s' \gets \textit{memo}[k]$\;
    }{
      $s' \gets \texttt{solve-defect}(k)$\;
      $\textit{memo}[k] \gets s'$\;
    }
  $g \gets$ the graph corresponding to the result returned from \texttt{solve-defect}$(k)$ with size $s'$\;
  $g\gets \texttt{EQC-Heu-Pro}(G, g, \gamma)$; \tcp{expansion}
  \Return{$|V(g)|$}\;
  }
\end{algorithm}

In this section, we present the omitted details of \texttt{EQC-Pro} mentioned in Section~\ref{sec:EQC-pro}. Based on two observations from Section~\ref{sec:EQC-pro}, namely that (1) solutions from \texttt{solve-defect}$(k)$ can be reused for multiple values of $s$ and that (2) each iteration yields a valid \texttt{kDC} solution which can be used to further strengthen the lower bound for \texttt{MaxEQC}. As shown in Algorithm~\ref{alg:improved-framework-memo}, \texttt{EQC-Pro} is an improved iterative algorithm built on the bottom-up doubling framework \texttt{EQC-BU}(Algorithm~\ref{alg:bottom-up-framework}) and the heuristic strategy \texttt{EQC-Heu-Pro} (Algorithm~\ref{alg:heu-extend}), and further incorporates two main techniques in the \texttt{MemoizedSearch} function: \emph{memoization} and \emph{expansion}.
First, we initiate a global memoization table (\texttt{memo}, Line 1) to cache results from previous calls to \texttt{solve-defect}$(k)$, thus avoiding redundant computations across iterations. 
Subsequently, in Lines 2-3, we compute an initial heuristic solution \( g \) using \texttt{Degen-Opt} and refine it with \texttt{EQC-Heu-Pro}, which attempts to expand the solution. The size of the resulting solution is used to initialize the lower bound \( lb \) of the search space for the bottom-up iterative framework.

After initialization (Line 4), \texttt{EQC-Pro} performs the bottom-up doubling iterative search in Lines 5-19, utilizing \texttt{MemoizedSearch} (Lines 20-28) in each iteration. We reuse cached results of \texttt{solve-defect}$(k)$ when available; otherwise, we invoke it and store the result in \texttt{memo} if its size is at least $s$. Moreover, \texttt{MemoizedSearch} further \emph{expands} the corresponding solution using \texttt{EQC-Heu-Pro}, which may yield a larger feasible solution. Once the iterative search terminates, the size of the maximum $\gamma$-EQC is returned in Line 19. The correctness of \texttt{EQC-Pro} directly follows from that of Algorithm~\ref{alg:bottom-up-framework}.
\begin{figure}[t]
    \vspace{-0.15in}
    \centering
    \subfigure[$\gamma = 0.96$]{
        \includegraphics[width=0.22\textwidth]{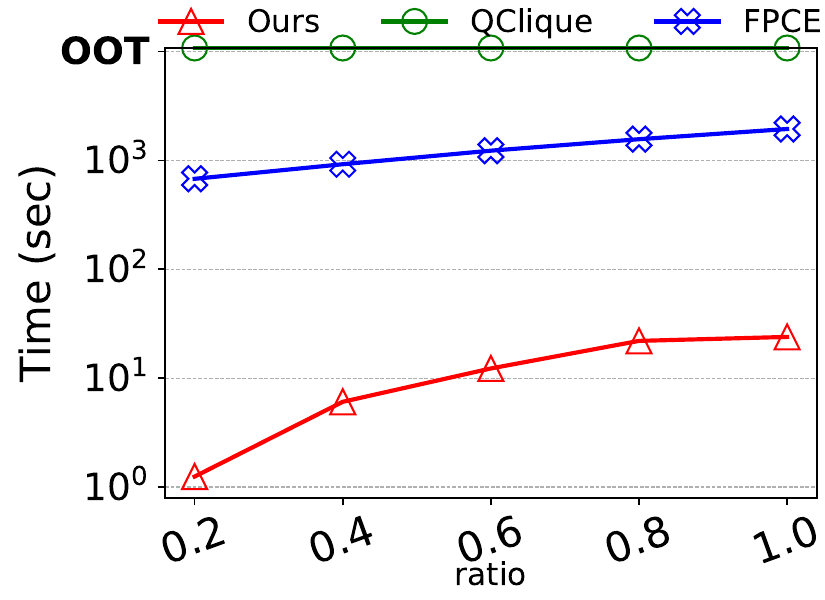}
        \label{fig:ratio096}
    }
    \subfigure[$\gamma = 0.98$]{
        \includegraphics[width=0.22\textwidth]{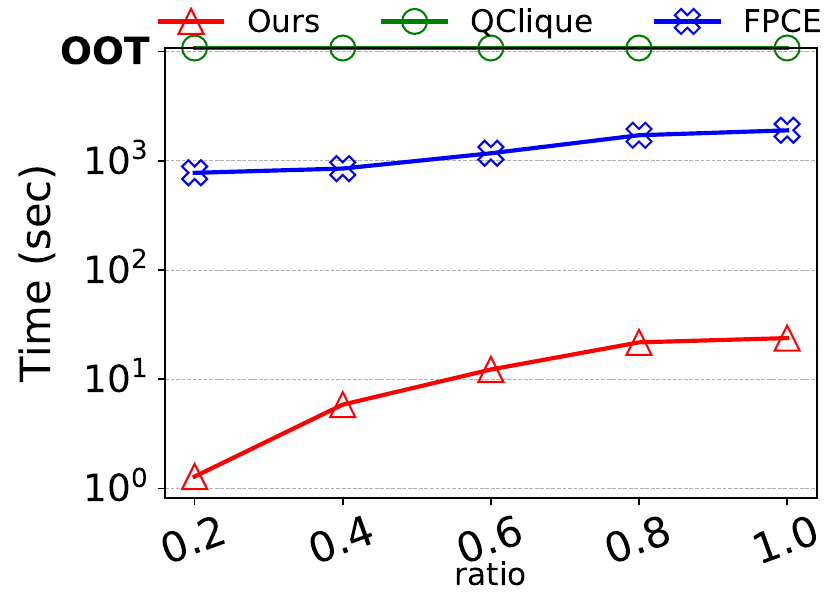}
        \label{fig:ratio098}
    }
    \vspace{-0.15in}
    \caption{Scalability test on instance \texttt{socfb-konect}.}
    \label{Fig:scalability}
\end{figure}

\section{Scalability Test} \label{sec:scale-appendix} To evaluate scalability, we select the largest instance across the two collections and perform subgraph sampling by randomly selecting vertices with increasing probabilities and extracting the induced subgraphs. The results are presented in Figure~\ref{Fig:scalability}. As shown, \texttt{QClique} consistently fails to solve any instance within the time limit, regardless of the sampling ratio or the value of $\gamma$, indicating its inability to handle large graphs. For \texttt{FPCE}, the runtime appears to grow nearly linearly with the number of vertices in the figure due to the use of a logarithmic scale on the $y$-axis. In reality, this behavior is consistent with its exponential time complexity of ${O}^*(2^n)$, where $n$ is the number of vertices in the sampled subgraph.

In contrast, the runtime of \texttt{EQC-Pro} initially increases more rapidly than that of \texttt{FPCE} as the graph size grows. This is because, for smaller subgraphs, the polynomial-time components, such as \texttt{Degen-Opt} and \texttt{EQC-Heu-Pro}, dominate the overall computational cost. However, as the graph size increases and the exponential search component becomes dominant, the advantage of our improved exponential base $O^*(\beta_\kappa^n)$ becomes increasingly apparent. Consequently, for larger graphs, \texttt{EQC-Pro} grows more slowly and consistently outperforms \texttt{FPCE} across all scales.

\end{document}